\documentclass[journal]{IEEEtran}

\usepackage{graphicx}
\usepackage[cmex10]{amsmath}
\usepackage{cases}
\usepackage{amsthm} 
\usepackage{cite}
\usepackage{amssymb}
\usepackage{algorithm}
\usepackage{algorithmic}
\usepackage{amsmath}
\usepackage{subeqnarray}
\usepackage{cases}
\usepackage{enumerate}
\usepackage{booktabs}
\usepackage{stfloats}
\usepackage{epstopdf}
\usepackage{float}
\usepackage{esint}
\usepackage{subfigure}
\usepackage{color}
\usepackage{mathrsfs}
\usepackage[final]{changes}

\newtheorem{theorem}{\hskip\parindent\bf{Theorem}}

\ifCLASSINFOpdf
\else
\fi

\begin{document}
	
	\title{Simultaneously Transmitting and Reflecting RIS (STAR-RIS) Assisted Multi-Antenna Covert Communications:  Analysis and Optimization}
	
	\author{Han Xiao,~\IEEEmembership{Student Member,~IEEE,} Xiaoyan Hu$^*$,~\IEEEmembership{Member,~IEEE,} \\
		Pengcheng Mu,~\IEEEmembership{Member,~IEEE,}   Wenjie Wang,~\IEEEmembership{Member,~IEEE,}  Tong-Xing Zheng,~\IEEEmembership{Member,~IEEE,} \\
		Kai-Kit~Wong,~\IEEEmembership{Fellow,~IEEE}, Kun~Yang,~\IEEEmembership{Fellow,~IEEE}

		\thanks{H. Xiao, X. Hu, P. Mu, W. Wang, and T.-X. Zheng are with the School of Information and Communication Engineering, Xi'an Jiaotong University, Xi'an 710049, China. (email: hanxiaonuli@stu.xjtu.edu.cn, xiaoyanhu@xjtu.edu.cn, \{pcmu, wjwang, zhengtx\}@mail.xjtu.edu.cn).}
		\thanks{K.-K. Wong is with the Department of Electronic and Electrical Engineering, University College London, London WC1E 7JE, U.K. (email: kai-kit.wong@ucl.ac.uk)}
		
		\thanks{K. Yang is with the School of Computer Science and Electronic Engineering, University of Essex, Colchester CO4 3SQ, U.K. (e-mail: kunyang@essex.ac.uk).}
	}

	\maketitle
	
	\begin{abstract}
		This paper investigates the multi-antenna covert communications assisted by a simultaneously transmitting and reflecting reconfigurable intelligent surface (STAR-RIS). In particular, to shelter the existence of covert communications between a multi-antenna transmitter and a single-antenna receiver from a warden, a friendly full-duplex receiver with two antennas is leveraged to make contributions where one antenna is responsible for receiving the transmitted signals and the other one transmits the jamming signals with a varying power to confuse the warden. Considering the worst case, the closed-form expression of the  minimum detection error probability (DEP) at the warden is derived and utilized in a covert constraint to guarantee the system performance. Then, we formulate an optimization problem maximizing the covert rate of the system under the covertness constraint and quality of service (QoS) constraint with communication outage analysis. To jointly design the active and passive beamforming of the transmitter and STAR-RIS, an iterative algorithm based on semi-definite relaxation (SDR) method and Dinkelbach's algorithm is proposed to effectively solve the non-convex optimization problem.
		Simulation results show that the proposed STAR-RIS-assisted scheme highly outperforms the case with conventional RIS, which validates the effectiveness of the proposed algorithm as well as the superiority of STAR-RIS in guaranteeing the covertness of wireless communications.
	\end{abstract}
	\begin{IEEEkeywords}
		Covert communication, STAR-RIS, multi-antenna, full-duplex, jamming. 
	\end{IEEEkeywords}
	
	\maketitle
	
	\vspace{-2mm}
	\section{Introduction}\label{sec:S1}
	With the advent of 5G era, people are becoming increasingly dependent on wireless communications driven by the advanced communications and data processing techniques. Massive important and sensitive information, e.g., ID information, confidential documents, etc., are transmitted over open wireless networks, which aggravates the eavesdropping risk.  Hence, people pay more and more attention to the problem of information security. Physical layer security (PLS) as a critical technology in protecting private information from eavesdropping attacks has drawn great attention in recent years \cite{A.Chorti_CSM22_Context,cui19, X.Hu_TVT17Onthe}. However, PLS techniques cannot perform well in scenarios with covertness requirements, e.g., secret military operations, since PLS is only able to protect the content information of wireless communications but is unable to hide the existence of communications \cite{zheng21, chen21}. 
	Recently, the technology of covert communications has emerged as a new security paradigm and attracted significant research interests in both civilian and military applications\cite{yan19}, which can shelter the existence of communications between transceivers  and provide a higher level of security for wireless communication systems.
	
	\vspace{-2mm}
	\subsection{Related Works}
	As a breakthrough work, \cite{bash13} first proved the fundamental limit of covert communications over additive white Gaussian noise (AWGN) channels from the perspective of information theory.
	It demonstrates that $O(\sqrt n)$ bits information can be transmitted covertly and reliably from transmitters to receivers over $n$ channel uses while the warden can achieve correct detections if the amount of transmitted information exceeds this square root law.
	Actually, this conclusion is pessimistic since the intrinsic uncertainty of wireless channels and the background noise are not taken into account in the considered communication systems of \cite{bash13}.
	For example, \cite{goeckel15} and \cite{wang18} indicate that $O(n)$ bits information can be transmitted to the receiver when  eavesdroppers don't exactly know the background noise power or the channel state information (CSI).
	Besides, existing works also resort to other uncertainties to enhance the performance of covert comunications\cite{Hu19, tao20, li20}, \cite{zheng21}. 
	In particular, a full-duplex receiver is  adopted in  \cite{Hu19} for generating power-varying artificial noise to obtain a decent covert rate. 
	In \cite{tao20}, the random transmit power is leveraged  to confuse the warden on the detection of covert transmissions.
	An uninformed jammer is introduced in \cite{li20} to assist covert communications by actively generating jamming signals  under different channel models. 
	Later in \cite{zheng21},  a multi-jammer scheme with uncoordinated jammer selection is studied to defeat the warden.
	Considering more practical scenarios, \cite{shahzad17} evaluates the influence of imperfect CSI on the system covert rate, and \cite{ma21} explores the case with multiple randomly distributed wardens and maximize the average effective covert throughput by jointly optimizing the transmit power and blocklength.
	
	The aforementioned works validate the effectiveness of the covert communication techniques from different perspectives, however, they only investigate the single-antenna covert communication scenarios. 
	In fact, multi-antenna technologies are beneficial in improving the capacity and reliability of traditional wireless communications which are also conductive to enhancing the performance of covert communications \cite{chen21, zheng19,  shahzad19}.
	Specifically, in \cite{chen21}, a multi-antenna transmitter and a full-duplex jamming receiver are utilized to alleviate the influence caused by the uncertainty of the warden.
	The authors in \cite{zheng19} study the potential performance gain of centralized and distributed multi-antenna transmitters in covert communication systems with random positions for wardens and interferers.
	Different from the above situations, a multi-antenna adversary warden is considered in \cite{shahzad19} which indicates that a slight increase in the antenna number of the adversary warden will result in a dramatical fall of covert rate.
	
	Although multi-antenna technologies can enhance covertness of communications through improving the ability of transmissions and receptions, it cannot tackle the issues brought by the randomness of wireless propagation environment.
	To break through this limitation, reconfigurable intelligent surface (RIS) has recently emerged as a promising solution \cite{wu19towards,X.HU_TCOM21RIS,lu20intelligent,zhou21intelligent,chen21enhancing,Wang21}.
	In particular, RIS is usually a two-dimensional metamaterial consisting of a large number of low-cost passive and adjustable reflecting elements. 
	The electromagnetic properties (e.g., phase and amplitude) of the signals impinged on RIS can be adaptively adjusted with the assistance of RIS elements via a smart controller. Hence, the utilization of RISs is capable of reshaping desirable wireless propagation environment,
	which has attracted intensive attentions and been leveraged in many wireless communication scenarios including covert communications \cite{lu20intelligent,zhou21intelligent,chen21enhancing,Wang21}.
	Specifically, \cite{lu20intelligent} generally summarizes the application potentials of RIS in improving covert communications. 
	Later in \cite{zhou21intelligent},  the authors explore the performance gain of covert communications provided by RIS and first prove that the perfect covertness can be achieved with the aid of RIS when the instantaneous CSI of the warden is available. A multi-input multi-output (MIMO) covert communication system assisted by RIS is applied in \cite{chen21enhancing} to resist the multi-antenna eavesdropper. 
	Also, \cite{Wang21} investigates the RIS-assisted multi-antenna covert communications by jointly optimizing the active and passive beamformers.
	
	\vspace{-3mm}
	\subsection{Motivation and Contributions}
	It is worth noting that the RISs applied by the aforementioned  works only reflect the incident signals which are limited to the scenarios that the transmitters and receivers locating at the same side of the RISs. 
	However, in practical cases, users may be on either side of RIS, and thus the flexibility and effectiveness of conventional RIS appear inadequate in these cases. To overcome this limitation, a novel technology called simultaneously transmitting and reflecting RIS (STAR-RIS) is further emerged \cite{liu21star}. 
	In particular, the incident signal will be separated into two parts when it arrives at the STAR-RIS, where one part is reflected to the same side of the incident signal and the other part is transited to the opposite side \cite{Liu22}.
	Note that STAR-RISs are capable of adjusting the reflected and transmitted signals by controlling the reflected and the transmitted coefficients simultaneously, which can help establish a more flexible full-space smart radio environment with $360^{\circ}$ coverage.
	Therefore, STAR-RIS possesses a huge application potential in wireless communications which has attracted intensive research interests from both academia and industry \cite{liu21star}. 
	However, the investigation of leveraging STAR-RISs into wireless communication systems is still in its infancy stage.
	As for secure communication systems, only a small number of state-of-the-art works have utilized STAR-RISs to enhance the system secure performance \cite{han22artificial,zhang22secrecy}.
	
	To our best knowledge, the application of STAR-RIS in covert communications has not been studied in existing works. This is the first work investigates a STAR-RIS assisted multi-antenna covert communication scenario so as to fully exploit the potentials of STAR-RIS in covert communications. 
	Our main contributions are summarized as follows:
	\vspace{-1mm}
	\begin{itemize}
		\item \emph{\textbf{STAR-RIS-assisted Covert Communication Architecture:}} A STAR-RIS-assisted covert communication architecture is constructed through which the legitimate users located on both sides of the STAR-RIS can be simultaneously served. Through elaborately design the reflected and transmitted coefficients of the STAR-RIS, this architecture can highly enhance the covert performance of the system though more flexible reconfigurations on the random wireless environment.
		\item \emph{\textbf{Closed-form Expressions for Covert System Indicators:}}
		Based on the constructed covert communication system, The closed-form expressions of the minimum detection error probability (DEP) and the corresponding optimal detection threshold at the warden are analytically derived considering the worst-case scenario.
		Based on a lower bound of the detection threshold and the large system analytic techniques, we further derive a lower bound of the average minimum DEP which is leveraged to jointly design the active and passive beamformers. The reasonability for choosing this lower bound is further validated by simulation results.
		\item 
		\emph{\textbf{Problem Formulation under Practical Constraints:}} We formulate an optimization problem aiming at maximizing the covert rate of the considered STAR-RIS-assisted covert communication system, under the covert communication constraint and the quality of service (QoS) constraint based on communication outage analysis, by jointly optimizing the active and passive beamforming at the base station (BS) and STAR-RIS.
		Due to the strongly coupled optimization variables and the characteristic amplitude constraint introduced by STAR-RIS,
		it is challenging to solve the formulated problem directly. 
		\item  \emph{\textbf{Alternating Algorithm with Guaranteed Convergence:}} An optimization algorithm based on alternating strategy is proposed to solve the formulated optimization problem in an iterative manner. In particular, the original problem is divided into three subproblems which are effectively solved by the semi-definite relaxation (SDR) method and Dinkelbach's algorithm. It is verified that the convergence of the proposed algorithm can always been guaranteed. 
		\item \emph{\textbf{Performance Improvement:}} The effectiveness of the proposed algorithm is validated by numerical results where we evaluate the average covert rate of the considered STAR-RIS-assisted system in comparison with a benchmark using the conventional RIS and a baseline algorithm called the globally convergent version of method of moving asymptotes (GCMMA). It is shown that the proposed algorithm can achieve great performance improvement compared with the baselines and the advantages are more obvious with a larger number of STAR-RIS elements.
	\end{itemize}
	
	The rest of this paper is organized as follows. In Section \ref{sec:S2}, we introduce the STAR-RIS-aided covert communication system model. The DEP of the warden and communication outage probability based on this model are derived and analyzed in Section \ref{sec:S3}. Section \ref{sec:S4} formulates the optimization problem and designs an iterative algorithm for jointly optimizing the passive and active beamforming. Section \ref{sec:S5} shows the simulation results to validate the effectiveness of the proposed algorithm. Finally, a conclusion is drawn in Section \ref{sec:S6}
	
	\textit{Notation:} Operator $\circ$ denotes the Hadamard product. $(\cdot)^T$, $(\cdot)^H$ and $(\cdot)^*$ represent transpose, conjugate transpose and conjugate, respectively. $\operatorname{Diag}(\mathbf{a})$ denotes a diagonal matrix with diagonal elements in vector $\mathbf{a}$, $\operatorname{diag}(\mathbf{A})$ denotes a vector whose elements are composed of the diagonal elements of matrix $\mathbf{A}$. $|\cdot|$ and $\|\cdot\|_2$ denote the complex modulus and the spectral norm, respectively. $\mathbb{C}^{N\times N}$ stands for the set of $N\times N$ complex matrices. $x\sim\mathcal{C N}(a,b)$ and $x~\sim\operatorname{exp}(\lambda)$ denote the circularly symmetric complex Gaussian random variable with mean $a$ and variance $b$ and the exponential random variable with mean $\lambda$, respectively. $\operatorname{Tr}(\cdot)$ represents the trace. $\mathbf{A}\succeq0$ indicates that matrix $\mathbf{A}$ is a positive semidefinite matrix. $\mathbf{I}_{N\times1}$ represents the vector with $N\times 1$ entries that are $1$.
	
	\vspace{-1mm}
	\section{System Model}\label{sec:S2}
		In this paper, we consider a STAR-RIS-assisted covert communication system model as shown in Fig.\ref{fig:System}, mainly consisting of a $M$-antenna BS transmitter (Alice) assisted by a STAR-RIS with $N$ elements, a covert user (Bob) and a  warden user (Willie) both equipped with a single antenna, and an assistant public user (Carol) with two antennas. Willie tries to detect the existence of data transmissions from Alice to Bob preparing for some security attacks.
	It is assumed that the single-antenna Bob and Willie work at the half-duplex mode,  while the two-antenna Carol operates in the full-duplex mode where one antenna receives the transmitted singles from Alice and the other one transmits jamming signals to weaken Willie's detection ability.
	Due to the existence of blockages, we assume that there is no direct links between Alice and all the users, which is reasonable in practical environment. 
	The STAR-RIS is deployed at the users' vicinity to enhance the end-to-end communications between Alice and the legal users Bob and Carol while confusing the detection of the warden user Willie.
	Without loss of generality, we consider a scenario that Bob and Carol locate on opposite sides of the STAR-RIS which can be served simultaneously by the reflected (T) and  transmitted (R) signals via STAR-RIS, respectively.\footnote{Similar to \cite{wu19towards} and  \cite{Wang21}, we ignore the signals reflected or transmitted more than once by the STAR-RIS considering  the severe path losses.}
	The energy splitting protocol is adopted for STAR-RIS whose all elements can work at T\&R modes simultaneously \cite{Liu22}.

	The wireless communication channels from Alice to STAR-RIS, and from STAR-RIS to Bob, Carol, Willie are denoted as \textbf{H$_{\rm{AR}}$}=$\sqrt{\textit{l$_{\rm{AR}}$}}\textbf{G$_{\rm{AR}}$}$$\in\mathbb{C^{\textit{N$\times M$} }}$ and \textbf{h}$_{\rm{rb}}$=$\sqrt{\textit{l$_{\rm{rb}}$}}\textbf{g$_{\rm{rb}}$}$$\in\mathbb{C^{\textit{N$\times 1$}}}$, \textbf{h}$_{\rm{rc}}$=$\sqrt{\textit{l$_{\rm{rc}}$}}\textbf{g$_{\rm{rc}}$}$$\in\mathbb{C^{\textit{N$\times 1$} }}$, \textbf{h}$_{\rm{rw}}$=$\sqrt{\textit{l$_{\rm{rw}}$}}\textbf{g$_{\rm{rw}}$}$$\in\mathbb{C^{\textit{N$\times 1$} }}$, respectively.
	In particular, \textbf{G$_{\rm{AR}}$} and \textbf{g$_{\rm{rb}}$}, \textbf{g$_{\rm{rc}}$}, \textbf{g$_{\rm{rw}}$} are the small-scale Rayleigh fading coefficients whose entries are independent identically distributed (i.i.d.) following the complex Gaussian distribution with zero mean and unit variance.
	In addition, \textit{l$_{\rm{AR}}$} and \textit{l$_{\rm{rb}}$}, \textit{l$_{\rm{rc}}$}, \textit{l$_{\rm{rw}}$} are the large-scale path loss coefficients modeled as $\sqrt{\frac{\rho_0}{d^\alpha}}$, where $\rho_0$ denotes the reference power gain at a distance of one meter ($m$), $\alpha$ represents the path-loss exponent, and $d$ corresponds to the node distances of \textit{d$_{\rm{AR}}$} and \textit{d$_{\rm{rb}}$}, \textit{d$_{\rm{rc}}$}, \textit{d$_{\rm{rw}}$}.
	We assume that the considered STAR-RIS-assisted covert communication system operates in the time division duplex (TDD) mode, so that the uplink channel estimation techniques based on STAR-RIS can be exploited to estimate the aforementioned communication CSI  by utilizing the channel reciprocity \cite{wu21}.
	As for the full-duplex  assistant user Carol, its self-interference channel can be modeled as $h_{\mathrm{cc}}=\sqrt{\phi}g_{\mathrm{cc}}$, where $g_{\mathrm{cc}}\sim \mathcal{C N}(0, 1)$, $\phi \in[0,1]$ is the self-interference cancellation (SIC) coefficient determined by the performing efficiency of the SIC \cite{Kim15, Hu19}.
		\begin{figure}[ht]
		\centering
		\includegraphics[scale=0.4]{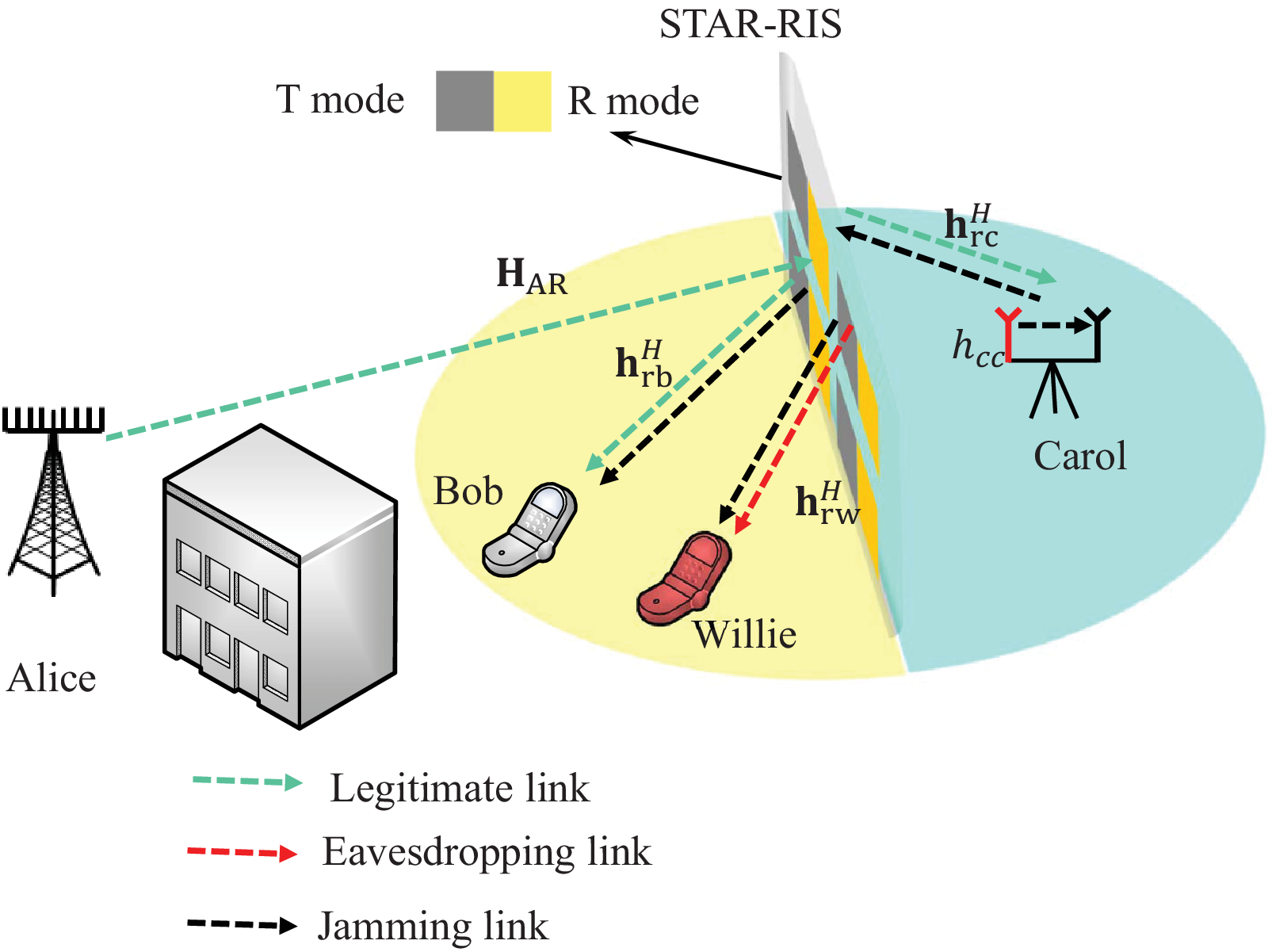}
		\caption{System model for STAR-RIS-assisted covert communications.}\label{fig:System}
	\end{figure}

	In this paper, we assume that the instantaneous CSI between STAR-RIS and Alice, Bob, Carol (i.e., $\mathbf{H}_{\mathrm{AR}}$, $\mathbf{h}_{\mathrm{rb}}$,  $\mathbf{h}_{\mathrm{rc}}$) is available at Alice, while only the statistical CSI between STAR-RIS and the Willie ($\mathbf{h}_{\mathrm{rw}}$) is known at Alice. In contrast, it is assumed that Willie is capable to know the instantaneous CSI of all the users, i.e., $\mathbf{h}_{\mathrm{rw}}$, $\mathbf{h}_{\mathrm{rb}}$ and $\mathbf{h}_{\mathrm{rc}}$, but can only access the statistical CSI of Alice, i.e., $\mathbf{H}_{\mathrm{AR}}$.
	As for the jamming signals transmitted by Carol, we assume that the power of the jamming signals, denoted as $P_\mathrm{j}$, follows the uniform distribution 
	with $P_{\mathrm{j}}^{\max }$ being the maximum power limit \cite{Hu19,Wang21}.
	It is assumed that Willie can only obtain the distribution of the jamming power, 
	and thus it is difficult for Willie to detect the existence of communications between Alice and Bob under the random jamming interference. 
	
	Under such a covert strategy, 
	the received signals at Bob and Carol in the considered STAR-RIS-assisted covert communication system can be respectively expressed as
	\begin{align}
		y_\mathrm{b}[k]=&\mathbf{h}_{\mathrm{rb}}^H \boldsymbol{\Theta}_\mathrm{r} \mathbf{H}_{\mathrm{AR}}\left(\mathbf{w}_\mathrm{b} s_\mathrm{b}[k]+\mathbf{w}_\mathrm{c}s_\mathrm{c}[k]\right)\notag\\
		&+\mathbf{h}_{\mathrm{rb}}^H \boldsymbol{\Theta}_\mathrm{t} \mathbf{h}_{\mathrm{rc}}^* \sqrt{P_\mathrm{j}} s_\mathrm{j}[k]+n_\mathrm{b}[k],\label{eq_rec_b}\\
		y_\mathrm{c}[k]=&\mathbf{h}_{\mathrm{rc}}^H \boldsymbol{\Theta}_\mathrm{t} \mathbf{H}_{\mathrm{AR}}\left(\mathbf{w}_\mathrm{b} s_\mathrm{b}[k]+\mathbf{w}_\mathrm{c}s_\mathrm{c}[k]\right)\notag\\
		&+h_{\mathrm{cc}} \sqrt{P_\mathrm{j}} s_\mathrm{j}[k]+n_\mathrm{c}[k],\label{eq_rec_c}
	\end{align}
	where $k \in \mathcal{K}\triangleq\{1, \ldots, K\}$ denotes the index of each communication channel use with the maximum number of $K$ in a time \mbox{slot}.
	$\boldsymbol{\Theta}_\mathrm{r}=\operatorname{Diag}\Big\{\sqrt{\beta_\mathrm{r}^1} e^{\mathrm{j} \phi_\mathrm{r}^1}, \ldots, \sqrt{\beta_\mathrm{r}^N}e^{\mathrm{j}\phi_\mathrm{r}^N}\Big\}$ and
	$\boldsymbol{\Theta}_\mathrm{t}=\operatorname{Diag}\Big\{\sqrt{\beta_\mathrm{t}^1} e^{\mathrm{j} \phi_\mathrm{t}^1}, \ldots, \sqrt{\beta_\mathrm{t}^N}e^{\mathrm{j}\phi_\mathrm{t}^N}\Big\}$
	respectively indicate the STAR-RIS reflected and transmitted coefficient matrices, where $\beta_\mathrm{r}^n, \beta_\mathrm{t}^n\in[0,1]$, $\beta_\mathrm{r}^n+\beta_\mathrm{t}^n=1$ and $\phi_\mathrm{r}^n, \phi_\mathrm{t}^n\in[0,2\pi)$, for $\forall n \in \mathcal{N} \triangleq\{1,2, \ldots, N\}$.
	In addition, {$\mathbf{w}_\mathrm{b}\in\mathbb{C^{\textit{M$\times 1$} }}$ and $ \mathbf{w}_\mathrm{c}\in\mathbb{C^{\textit{M$\times 1$} }}$ are the precoding vectors at Alice for Bob and Carol, respectively.
		$s_\mathrm{b}[k]$ and $s_\mathrm{c}[k]$ $\sim \mathcal{C N}(0, 1)$ are the signals transmitted by Alice to Bob and Carol while $s_\mathrm{j}[k]$$\sim \mathcal{C N}(0, 1)$ is the jamming signal transmitted by Carol, where we try to hide the transmission of $s_\mathrm{b}[k]$ from the detection of Willie.
		Also, we use $n_\mathrm{b}[k]\sim\mathcal{C N}(0, \sigma_\mathrm{b}^2)$ and $n_\mathrm{c}[k]\sim\mathcal{C N}(0, \sigma_\mathrm{c}^2)$ to represent the AWGN noise received at Bob and Carol with $\sigma_\mathrm{b}^2$ and $\sigma_\mathrm{c}^2$ being the corresponding noise power.\footnote{It is worth noting that we ignore the jamming signals reflected by STAR-RIS at Carol mainly due to the fact that it is negligible compared with the self-interference  jamming signals received by Carol.}

	Note that the instantaneous jamming power $P_\mathrm{j}$ and the self-interference channel $h_{\mathrm{cc}}$ of Carol are unavailable at Alice, thus the randomness introduced by $P_\mathrm{j}$ and $h_{\mathrm{cc}}$ is possible to incur communication outage. When the required communication {rate from Alice to Bob ($R_\mathrm{b}$) or Carol ($R_\mathrm{c}$) exceeds the corresponding channel capacity ($C_\mathrm{b}$, $C_\mathrm{c}$), 
		the communication outage occurs.
		Hence, the communication outage probability at Bob and Carol can be respectively expressed as
		\vspace{-1mm}
		\begin{align}
			\delta_\mathrm{AB}=&\operatorname{Pr}\left(C_\mathrm{b}<R_\mathrm{b}\right),\label{eq_outage_b}\\
			\delta_\mathrm{AC}=&\operatorname{Pr}\left(C_\mathrm{c}<R_\mathrm{c}\right). \label{eq_outage_c}
		\end{align}
		
		\vspace{-2mm}To guarantee the communication quality between Alice and Bob/Carol, the issue of communication outage also needs to be addressed in the considered covert communication system.
		
		\section{Analysis on STAR-RIS-Assisted Covert Communications}\label{sec:S3}%
		\subsection{Covert Communication Detection Strategy at Willie}\label{sec:S3_P1}
		In this section, we detail the detection strategy of Willie for STAR-RIS-assisted covert communications from Alice to Bob. In particular, Willie attempts to judge whether there exists covert transmissions based on the received signal sequence $\{y_\mathrm{w}[k]\}_{k \in \mathcal{K}}$ in a time slot.
		Thus, Willie has to face a binary hypothesis for detection, which includes a null hypothesis, $\mathcal{H}_0$, representing that Alice only transmits public signals to Carol, and an alternative hypothesis, $\mathcal{H}_1$, indicating that Alice transmits both public signals and covert signals to Coral and Bob, respectively.
		Furthermore, the received signals at Willie based on the two hypotheses are given by
		\begin{align}
			\mathcal{H}_0: y_\mathrm{w}[k]=&\mathbf{h}_{\mathrm{rw}}^H \boldsymbol{\Theta}_\mathrm{r} \mathbf{H}_{\mathrm{AR}} \mathbf{w}_\mathrm{c}s_\mathrm{c}[k]+\mathbf{h}_{\mathrm{rb}}^H \boldsymbol{\Theta}_\mathrm{t} \mathbf{h}_{\mathrm{rc}}^* \sqrt{P_\mathrm{j}} s_\mathrm{j}[k]\notag\\&+n_\mathrm{w}[k], ~k\in \mathcal{K}, \label{eq_hypoH0_w}\\
			\mathcal{H}_1: y_\mathrm{w}[k]=&\mathbf{h}_{\mathrm{rw}}^H \boldsymbol{\Theta}_\mathrm{r} \mathbf{H}_{\mathrm{AR}} \mathbf{w}_\mathrm{b} s_\mathrm{b}[k]+\mathbf{h}_{\mathrm{rw}}^H \boldsymbol{\Theta}_\mathrm{r} \mathbf{H}_{\mathrm{AR}} \mathbf{w}_\mathrm{c} s_\mathrm{c}[k]\notag\\&+\mathbf{h}_{\mathrm{rb}}^H \boldsymbol{\Theta}_\mathrm{t} \mathbf{h}_{\mathrm{rc}}^* \sqrt{P_\mathrm{j}}s_\mathrm{j}[k]+n_\mathrm{w}[k], ~k\in \mathcal{K},\label{eq_hypoH1_w}
		\end{align}
		where $n_\mathrm{w}[k]\sim\mathcal{C N}(0, \sigma_\mathrm{w}^2)$ is the AWGN received at Willie. 
		We assume that Willie utilizes a radiometer to detect the  covert signals from Alice to Bob, owing to its properties of low
		complexity and ease of implementation \cite{zhou19, tao20}.
			\begin{figure*}[t]
			\hrulefill
			\vspace*{4pt}\\
			\begin{equation}\label{eq_repower_w}
				\overline{P}_\mathrm{w}=\lim _{K \rightarrow+\infty} \frac{1}{K} \sum_{k=1}^K\left|y_\mathrm{w}[k]\right|^2=\begin{cases}\left|\mathbf{h}_{\mathrm{rw}}^H \boldsymbol{\Theta}_\mathrm{r} \mathbf{H}_{\mathrm{AR}} \mathbf{w}_\mathrm{c}\right|^2+\left|\mathbf{h}_{\mathrm{rw}}^H \boldsymbol{\Theta}_\mathrm{t} \mathbf{h}_{\mathrm{rc}}^*\right|^2 P_\mathrm{j}+\sigma_\mathrm{w}^2, & \mathcal{H}_0, \\ \left|\mathbf{h}_{\mathrm{rw}}^H \boldsymbol{\Theta}_\mathrm{r} \mathbf{H}_{\mathrm{AR}} \mathbf{w}_\mathrm{b}\right|^2+\left|\mathbf{h}_{\mathrm{rw}}^H \boldsymbol{\Theta}_\mathrm{r} \mathbf{H}_{\mathrm{AR}} \mathbf{w}_\mathrm{c}\right|^2+\left|\mathbf{h}_{\mathrm{rw}}^H \boldsymbol{\Theta}_\mathrm{t} \mathbf{h}_{\mathrm{rc}}^*\right|^2 P_\mathrm{j}+\sigma_\mathrm{w}^2, & \mathcal{H}_1,\end{cases}
			\end{equation}
		\end{figure*}
	
		According to the working mechanism of the radiometer, the average power of the received signals at Willie in a time slot, i.e., $\overline{P}_\mathrm{w}=\frac{1}{K} \sum_{k=1}^K\left|y_\mathrm{w}[k]\right|^2$, is employed for statistical test. Similar to the existing works, (e.g., \cite{Hu19,Wang21, zheng21}), we assume that Willie uses infinite number of signal samples to implement binary detection, i.e., $K \rightarrow \infty$.
		Hence, the average received power at Willie $\overline{P}_\mathrm{w}$ can be asymptotically approximated as (\ref{eq_repower_w}) which is shown at the top of the next page.
		Willie needs to analyze $\overline{P}_\mathrm{w}$ to decide whether the communication between Alice and Bod is under the  hypotheses of $\mathcal{H}_0$ or $ \mathcal{H}_1$, and its decision rule can be presented as
		\begin{equation}\label{eq_deci_w}
			\overline{P}_\mathrm{w} \underset{\mathcal{D}_0}{\stackrel{\mathcal{D}_1}{\gtrless}} \tau_\mathrm{dt},
		\end{equation}
		where $\mathcal{D}_0$ (or $\mathcal{D}_1$) indicates the decision that Willie favors $\mathcal{H}_0$ (or $\mathcal{H}_1$), and $\tau_\mathrm{d t}>0$ is the corresponding detection threshold.
	
		In this paper, we adopt the DEP as Willie's detection performance metric and consider the worst case scenario that Willie can optimize its detection threshold to obtain the minimum DEP.
		According to the Neyman-Pearson criterion, the minimum DEP of Willie is the likelihood ratio \cite{Wang21, li20}, which can be expressed as 
		\begin{align}\label{eq_likelihood_w}
			\Lambda\left(\mathbf{y}_\mathrm{w}\right)=\frac{f_{\mathbf{y}_\mathrm{w} \mid \mathcal{H}_1}\left(\mathbf{y}_\mathrm{w} \mid \mathcal{H}_1\right)}{f_{\mathbf{y}_\mathrm{w} \mid \mathcal{H}_0}\left(\mathbf{y}_\mathrm{w} \mid \mathcal{H}_0\right)},
		\end{align}
		where $\mathbf{y}_\mathrm{w}=\left\{y_\mathrm{w}[1], \ldots, y_\mathrm{w}[K]\right\}$ is Willie's received signal vector, $f_{\mathbf{y}_\mathrm{w} \mid \mathcal{H}_0}$ and $f_{\mathbf{y}_\mathrm{w} \mid \mathcal{H}_1}$ are the probability density functions (PDFs) of the sampling signals when $\mathcal{H}_0$ and $\mathcal{H}_1$ are true, respectively.
		However, it is difficult to derive the likelihood ratio due to the fact that the instantaneous CSI of $ \mathbf{H}_{\mathrm{AR}} $ is unavailable at Willie, which introduces extra randomness in received signal $\mathbf{y}_\mathrm{w}$.
		
		To solve this problem, we will derive the false alarm (FA) probability and the miss detection (MD) probability to further obtain the minimum DEP.
		Specifically, FA probability represents the probability of Willie making the decision $\mathcal{D}_1$ under $\mathcal{H}_0$, i.e., $P_{\mathrm{FA}}=\operatorname{Pr}\left(\mathcal{D}_0 \mid \mathcal{H}_1\right)$, while MD probability is the probability of Willie making the decision $\mathcal{D}_0$ under $\mathcal{H}_1$, i.e.,    $P_{\mathrm{MD}}=\operatorname{Pr}\left(\mathcal{D}_1 \mid \mathcal{H}_0\right)$. Hence, the DEP of Willie can be written as
		\begin{equation}\label{eq_DEP}
			P_\mathrm{e}=P_{\mathrm{FA}}+P_{\mathrm{MD}}=\operatorname{Pr}\left(\mathcal{D}_0 \mid \mathcal{H}_1\right)+\operatorname{Pr}\left(\mathcal{D}_1 \mid \mathcal{H}_0\right).
		\end{equation}
		
		It is easy to note that $0 \leq P_\mathrm{e} \leq 1$ where $P_\mathrm{e}=0$ indicates that Willie can always correctly detect the existence of the cover communications between Alice and Bob, while $P_\mathrm{e}=1$ means that Willie is never able to make correct detections.
		By choosing a reasonable detection threshold $\tau_\mathrm{d t}$, the minimum DEP denoted as $P_\mathrm{e}^*$, can be obtained at Willie. Also, in order to achieve the covertness of communications, it is necessary to guarantee $P_\mathrm{e}^* \geq 1-\epsilon$, where $\epsilon \in(0,1)$ is a quite small value required by the system performance indicators.

		\subsection{Analysis on Detection Error Probability}\label{sec:S3_P2}
		In this section, we first derive the analytical expressions for $P_{\mathrm{FA}}$ and $P_{\mathrm{MD}}$ in closed form, based on the distribution of $\overline{P}_\mathrm{w}$ under $\mathcal{H}_0$ and $\mathcal{H}_1$.
		Then, the optimal detection threshold $\tau_\mathrm{d t}^*$ and the minimum DEP $P_\mathrm{e}^*$ are obtained by analyzing the analytical expression of DEP.
		Considering the fact that the instantaneous CSI of channel $\mathbf{h}_\mathrm{rw}$ is not available at Alice, thus we derive a lower bound for  the average the minimum DEP over $\mathbf{h}_{\mathrm{rw}}$. In particular, the analytical expressions for $P_{\mathrm{FA}}$ and $P_{\mathrm{MD}}$ can be obtained through Theorem \ref{th1}.
		\begin{theorem}\label{th1}
			The analytical expressions for FA probability $P_\mathrm{FA}$ and MD probability $P_\mathrm{MD}$ are respectively given as
			\begin{equation}\label{eq_P_FA}
				\begin{aligned}
					&P_\mathrm{FA}=\\ &\begin{cases}1, & \tau_\mathrm{d t}<\sigma_\mathrm{w}^2, \\
						1-\frac{\left(\tau_\mathrm{d t}-\sigma_\mathrm{w}^2\right)+\lambda e^{-\frac{\tau_\mathrm{d t}-\sigma_\mathrm{w}^2}{\lambda}}-\lambda}{\gamma P_\mathrm{j}^{\max} }, &\sigma_\mathrm{w}^2 \leq \tau_\mathrm{d t}<\sigma_\mathrm{w}^2+\gamma P_\mathrm{j}^{\max} , \\
						\frac{e^{-\frac{\tau_\mathrm{d t}-\sigma_\mathrm{w}^2}{\lambda}}\left(e^{\frac{\gamma P_\mathrm{j}^{\max}}{\lambda}}-1\right) \lambda}{\gamma P_\mathrm{j}^{\max}}, &\tau_\mathrm{d t} \geq  \sigma_\mathrm{w}^2+\gamma P_\mathrm{j}^{\max}, \end{cases}
				\end{aligned}
			\end{equation}
			\begin{equation}\label{eq_P_MD}
				\begin{aligned}
					&P_\mathrm{MD}=\\&\begin{cases}
						0, & \tau_\mathrm{d t}<\sigma_\mathrm{w}^2, \\
						\frac{\left(\tau_\mathrm{d t}-\sigma_\mathrm{w}^2\right)+\tilde{\lambda} e^{-\frac{\tau_\mathrm{d t}-\sigma_\mathrm{w}^2}{\tilde{\lambda}}}-\tilde{\lambda}}{\gamma P_\mathrm{j}^{\max}},& \sigma_\mathrm{w}^2 \leq \tau_\mathrm{d t}<\sigma_\mathrm{w}^2+\gamma P_\mathrm{j}^{\max}, \\
						1-\frac{\mathrm{e}^{-\frac{\tau_\mathrm{d t}-\sigma_\mathrm{w}^2}{\tilde{\lambda}}}\left(e^{\frac{\gamma P_\mathrm{j}^{\max}}{\tilde{\lambda}}}-1\right) \tilde{\lambda}}{\gamma P_\mathrm{j}^{\max}}, &\tau_\mathrm{d t}\geq \gamma P_\mathrm{j}^{\max} +\sigma_\mathrm{w}^2,  \end{cases}	
				\end{aligned}
			\end{equation}
			which are shown in closed form with $\lambda=\left\|\mathbf{h}_{\mathrm{rw}}^H \boldsymbol{\Theta}_\mathrm{r}\right\|_2^2\mathbf{w}_\mathrm{c}^H \mathbf{w}_\mathrm{c}$,
			$\tilde{\lambda}=\left\|\mathbf{h}_{\mathrm{rw}}^H \boldsymbol{\Theta}_\mathrm{r}\right\|_2^2\big(\mathbf{w}_\mathrm{b}^H\mathbf{w}_\mathrm{b}+ $ $\mathbf{w}_\mathrm{c}^H \mathbf{w}_\mathrm{c}\big)$, and $\gamma=\left|\mathbf{h}_{\mathrm{rw}}^H \boldsymbol{\Theta}_\mathrm{t} \mathbf{h}_{\mathrm{rc}}^*\right|^2$.
		\end{theorem}
		
		\begin{proof}
			The proof is given in Appendix \ref{Apd_A}.
		\end{proof}
		
		According to \eqref{eq_P_FA} and \eqref{eq_P_MD}, we find that $\sigma_\mathrm{w}^2$ and $\sigma_\mathrm{w}^2+\gamma P_\mathrm{j}^{\max} $ are two important boundaries affecting the values of $P_\mathrm{FA}$ and $P_\mathrm{MD}$.
		Specifically, when $\tau_\mathrm{dt}\leq\sigma_\mathrm{w}^2$, complete FAs will be performed by Willie and the MDs can be totally avoided.
		Furthermore, with the increase of detection threshold $\tau_\mathrm{dt}$ from $\sigma_\mathrm{w}^2$ to $+\infty$, $P_\mathrm{FA}$ will experience a decrease from 1 to 0, while $P_\mathrm{MD}$ will have an opposite trend.
	Based on the analytical expressions of \eqref{eq_P_FA} and \eqref{eq_P_MD}, the analytical closed-form expressions of DEP can be given as \eqref{eq_DEP_expression} which is shown at the top of the next page.

		\begin{figure*}[t]
		\hrulefill
		\begin{equation}\label{eq_DEP_expression}
			P_\mathrm{e}=\begin{cases}
				1, & \tau_\mathrm{d t}<\sigma_\mathrm{w}^2, \\
				1+\frac{\tilde{\lambda}\left(e^{-\frac{\tau_\mathrm{d t}-\sigma_\mathrm{w}^2}{\tilde{\lambda}}}-1\right)-\lambda e^{-\frac{\tau_\mathrm{d t}-\sigma_\mathrm{w}^2}{\lambda}}+\lambda}{\gamma P_\mathrm{j}^{\max}}, & \sigma_\mathrm{w}^2 \leq \tau_\mathrm{d t}<\sigma_\mathrm{w}^2+\gamma P_\mathrm{j}^{\max}, \\
				1+\frac{\tilde{\lambda} e^{-\frac{\tau_\mathrm{d t}-\sigma_\mathrm{w}^2}{\tilde{\lambda}}}\left(1-e^{\frac{\gamma P_\mathrm{j}^{\max}}{\tilde{\lambda}}}\right)+\lambda e^{-\frac{\tau_\mathrm{d t}-\sigma_\mathrm{w}^2}{\lambda}}\left(e^{\frac{\gamma P_\mathrm{j}^{\max}}{\lambda}}-1\right)}{\gamma P_\mathrm{j}^{\max}},&\tau_\mathrm{d t}\geq\gamma P_\mathrm{j}^{\max}+\sigma_\mathrm{w}^2 ,
			\end{cases}
		\end{equation}
	\end{figure*}
	
	It is important to point out that we consider the worst case scenario that Willie can optimize its detection threshold to minimize the DEP. Here, the closed-form solution of the optimal $\tau_\mathrm{dt}^*$ is provided in the following theorem.
	\begin{theorem}\label{th2}
		The optimal detection threshold $\tau_\mathrm{dt}^*$ to minimize the DEP of Willie in the considered STAR-RIS-assisted communication system is given by
		\vspace{-2mm}
		\begin{equation}\label{eq_opt_DT}
			\tau_\mathrm{dt}^*=\frac{\tilde{\lambda} \lambda}{\tilde{\lambda}-\lambda} \ln \Delta+\sigma_\mathrm{w}^2 \in \left[\sigma_\mathrm{w}^2+\gamma P_\mathrm{j}^{\max},\infty\right),
		\end{equation}
		where $\Delta=\frac{e^{\frac{\gamma P_\mathrm{j}^{\max}}{\lambda}}-1}{e^{\frac{\gamma P_\mathrm{j}^{\max} }{\tilde{\lambda}}}-1}$ is a function of $\lambda$, $\tilde{\lambda}$ and $\gamma$.
	\end{theorem}
	\begin{proof}
		The proof is given in Appendix \ref{Apd_B}. 
	\end{proof}
	
	Substituting \eqref{eq_opt_DT} into \eqref{eq_DEP_expression} and adopting some algebraic manipulations, the analytical closed-form expression of the minimum DEP can be obtained as
	\begin{align}\label{eq_opt_DEP}
		P_\mathrm{e}^*=&
		1-\notag \\
		&\frac{\tilde{\lambda}\left(e^{\frac{\gamma P_\mathrm{j}^{\max}}{\tilde{\lambda}}}-1\right)(\Delta)^{\frac{\lambda}{\lambda-\tilde{\lambda}}}-\lambda\left(e^{\frac{\gamma P_\mathrm{j}^{\max}}{\lambda}}-1\right)(\Delta)^{\frac{\tilde{\lambda}}{\lambda-\tilde{\lambda}}}}{\gamma P_\mathrm{j}^{\max}}.	
	\end{align}
	
	\vspace{-2mm}Since Alice only knows the statistical CSI of channel $\mathbf{h}_\mathrm{rw}$, the average minimum DEP over $\mathbf{h}_\mathrm{rw}$, denoted as $\overline{P}_\mathrm{e}^*=\mathbb{E}_{\mathbf{h}_{\mathrm{rw}}}\left(P_\mathrm{e}^*\right)$, is usually utilized to evaluate the covert communications between Alice and Bob \cite{zhou19, chen21}.
	In \eqref{eq_opt_DEP}, $\lambda$, $\tilde{\lambda} $ and $\gamma $ are all random variables including $\mathbf{h}_\mathrm{rw}$, and thus they are coupled to each other, which makes it challenging to calculate the $\overline{P}_\mathrm{e}^*$ directly.
	To solve this problem, large system analytic techniques are utilized to remove the couplings among $\lambda$, $\tilde{\lambda}$ and $\gamma$, which are widely adopted to  analyze the performance limitation of wireless communication systems (e.g., \cite{evans00,wu16, Wu20, ding20, Wang21}).
	By assuming that the STAR-RIS is equipped with a large number of low-cost elements, then we can obtain the asymptotic analytic result of $P_\mathrm{e}^*$.
	In particular, we first apply the large system analytic technique on $\lambda$,  then the asymptotic equality about $\lambda$ can be given as 
	\begin{equation}\label{eq_asy_lambda_1}
		\begin{aligned}[b]
			\lim _{N \rightarrow \infty} \frac{\left\|\mathbf{h}_{\mathrm{rw}}^H \boldsymbol{\Theta}_\mathrm{r}\right\|_2^2\varpi_\mathrm{c}}{N}  &=\lim _{N \rightarrow \infty} \frac{\mathbf{h}_{\mathrm{rw}}^H \boldsymbol{\Theta}_\mathrm{r} \boldsymbol{\Theta}_\mathrm{r}^H \mathbf{h}_{\mathrm{rw}}\varpi_\mathrm{c}}{N} \\
			&\stackrel{(a)}{\rightarrow} \frac{l_{\mathrm{rw}} \varpi_\mathrm{c}}{N} \operatorname{tr}\left(\boldsymbol{\Theta}_\mathrm{r} \boldsymbol{\Theta}_\mathrm{r}^H\right) \\
			&=\frac{l_{\mathrm{rw}}\varpi_\mathrm{c}}{N} \theta_\mathrm{r}=\frac{\lambda_\mathrm{a}}{N} ,
		\end{aligned}
	\end{equation}
	where the convergence (a) is due to \cite[Corollary 1]{evans00}.
	Here, $\varpi_\mathrm{c}=\mathbf{w}_\mathrm{c}^H\mathbf{w}_\mathrm{c}$, $\theta_\mathrm{r}=\operatorname{diag}(\boldsymbol{\Theta}_\mathrm{r})^H\operatorname{diag}(\boldsymbol{\Theta}_\mathrm{r})$,  and $\lambda_\mathrm{a}=l_{\mathrm{rw}}\varpi_\mathrm{c} \theta_\mathrm{r}$ is the asymptotic result of $\lambda$. Similarly, the asymptotic result  of $\tilde{\lambda} $ can be expressed as
	\begin{equation}\label{eq_asy_lambda_2}
		\tilde{\lambda}_\mathrm{a}=l_{\mathrm{rw}} \theta_\mathrm{r}(\varpi_\mathrm{b}+\varpi_\mathrm{c}),
	\end{equation}
	where we define $\varpi_\mathrm{b}=\mathbf{w}_\mathrm{b}^H\mathbf{w}_\mathrm{b}$.
	
	With the results of $\lambda_\mathrm{a}$ and $\tilde{\lambda}_\mathrm{a}$ based on the large system analytic technique, the uncertainty of $\lambda$ and $\tilde{\lambda}$ can be removed from the perspective of Alice. Substituting $\lambda_\mathrm{a}$ and $\tilde{\lambda}_\mathrm{a}$ into \eqref{eq_opt_DEP}, we obtain the asymptotic analytical result of the minimum DEP $P_\mathrm{e}^*$ with respect to (w.r.t.) the random variable $\gamma$
	\begin{equation}\label{eq_asy_DEP}
		P_\mathrm{ea}^*=1-\frac{l_{\mathrm{rw}}\theta_\mathrm{r}\varpi_\mathrm{b}}{\gamma P_\mathrm{j}^{\max}}(\Delta(\gamma))^{\frac{-\varpi_\mathrm{c}}{\varpi_\mathrm{b}}}\left(e^{\frac{\gamma P_\mathrm{j}^{\max}}{l_{\mathrm{rw}}\theta_\mathrm{r}(\varpi_\mathrm{b}+\varpi_\mathrm{c})}}-1\right).	\
	\end{equation}
	
	It is easy to verify that $\gamma=\left|\mathbf{h}_{\mathrm{rw}}^H \boldsymbol{\Theta}_\mathrm{t} \mathbf{h}_{\mathrm{rc}}^*\right|^2 \sim\exp(\lambda_\mathrm{rw})$ where $\lambda_\mathrm{rw}=\left\|\boldsymbol{\Theta}_\mathrm{t}\mathbf{h}_{\mathrm{rc}}^*\right\|^2_2$.
	By averaging $P^*_\mathrm{ea}$ over $\gamma$, we can get the average asymptotic analytical result of the minimum DEP as
	\begin{equation}\label{eq_avera_asy_DEP}
		\begin{aligned}
			\overline{P}_\mathrm{e a}^*=& \mathbb{E}_\gamma\left(P_\mathrm{e a}^*\right) \\
			=& \int_0^{+\infty}\Bigg(1-\frac{l_{\mathrm{rw}} \theta_\mathrm{r} \varpi_\mathrm{b}}{\gamma P_\mathrm{j}^{\max}}(\Delta(\gamma))^{\frac{-\varpi_\mathrm{c}}{\varpi_\mathrm{b}}}\\
			&\left(e^{\frac{\gamma P_\mathrm{j}^{\max}}{l_{\mathrm{rw}} \theta_\mathrm{r}\left(\varpi_\mathrm{b}+\varpi_\mathrm{c}\right)}}-1\right)\Bigg)
			\times \frac{1}{\lambda_{\mathrm{rw}}} e^{-\frac{\gamma}{\lambda_{\mathrm{rw}}}}d \gamma.
		\end{aligned}
	\end{equation}
	Due to the existence of $\Delta(\gamma)$ in $P_\mathrm{ea}^*$, the integral in \eqref{eq_avera_asy_DEP} for calculating $\overline{P}_\mathrm{ea}^*$ over the random variable $\gamma$  is non-integrable. Therefore, the exact analytical expression for $\overline{P}_\mathrm{ea}^*$ is mathematically intractable.
	In order to guarantee the covert constraint $\overline{P}_\mathrm{ea}^*\geq 1-\epsilon$  always holds,
	we further adopt a lower bound of $\overline{P}_\mathrm{ea}^*$ to evaluate the covertness of communications.
	
	Specifically, we use a lower bound $\hat\Delta(\gamma)\triangleq e^{ \gamma P_\mathrm{j}^{\max}\left(\frac{\tilde{\lambda}_\mathrm{a}-\lambda_\mathrm{a}}{\tilde{\lambda}_\mathrm{a} \lambda_\mathrm{a}}\right)}$ to replace $\Delta(\gamma)$.
	It is easy to demonstrate that $\Delta(\gamma)>\hat\Delta(\gamma)$  always holds when $\gamma\in(0,+\infty)$, and the relative gap between these two variables gradually goes to zero with the increase of $\gamma$.\footnote{The reasonability about the selection of $\hat\Delta(\gamma)$ (or $\hat{P}_\mathrm{e a}^*$) will be further verified based on the simulation results.}
	Replacing the $\Delta(\gamma)$ in \eqref{eq_avera_asy_DEP} with $\hat\Delta(\gamma)$, we can get a lower bound of $\overline{P}_\mathrm{ea}^*$ given by
	\begin{equation}\label{eq_low_avera_asy_DEP}
		\begin{aligned}[b]
			\hspace{-2mm}\hat{P}_\mathrm{e a}^*\triangleq &\int_0^{+\infty} \Bigg(1-\frac{l_{\mathrm{rw}} \theta_\mathrm{r} \varpi_\mathrm{b}}{\gamma P_\mathrm{j}^{\max}}\left(\hat\Delta(\gamma)\right)^{\frac{-\varpi_\mathrm{c}}{\varpi_\mathrm{b}}}\\
			&\left(e^{\frac{\gamma P_\mathrm{j}^{\max}}{l_{\mathrm{rw}}\theta_\mathrm{r}\left(\varpi_\mathrm{b}+\varpi_\mathrm{c}\right)}}-1\right)\Bigg)
			\times\frac{1}{\lambda_{\mathrm{rw}}} e^{-\frac{\gamma}{\lambda_{\mathrm{rw}}}}d \gamma\\
			=&1+\frac{l_{\mathrm{rw}} \theta_\mathrm{r} \varpi_\mathrm{b}\left(\ln \frac{l_{\mathrm{rw}} \theta_\mathrm{r}\left(\varpi_\mathrm{b}+\varpi_\mathrm{c}\right)}{l_{\mathrm{rw}} \theta_\mathrm{r}\left(\varpi_\mathrm{b}+\varpi_\mathrm{c}\right)+P_\mathrm{j}^{\max} \lambda_{\mathrm{rw}}}\right)}{P_\mathrm{j}^{\max} \lambda_{\mathrm{rw}}} <  \overline{P}_{\mathrm{ea}}^*.
		\end{aligned}
	\end{equation}
	
	Therefore, in the following sections, $\hat{P}_\mathrm{e a}^*\geq 1-\epsilon$ will be leveraged as a tighter covert constraint to jointly design the active and passive
	beamforming variables of the system.
	
	\subsection{Analysis on Communication Outage Probability}
	As we mentioned before, the randomness introduced by the jamming signal power $P_\mathrm{j}$ and the self-interference channel $h_{\mathrm{c}c}$ of  Carol is possible to result in communication outages between Alice and Bob/Carol. In order to guarantee the QoS of communications, the communication outage constraints should be taken into consideration.
	
	It is known that the channel capacity at Bob and Carol can be respectively written as 
	\begin{align}
		C_\mathrm{b}=&\notag\log _2 \Bigg(1+\\
		&\frac{\left|\mathbf{h}_{\mathrm{rb}}^H \boldsymbol{\Theta}_\mathrm{r}\mathbf{H}_{\mathrm{AR}} \mathbf{w}_ \mathrm{b}\right|^2}{\left|\mathbf{h}_{\mathrm{rb}}^H \boldsymbol{\Theta}_\mathrm{r} \mathbf{H}_{\mathrm{AR}} \mathbf{w}_ \mathrm{c}\right|^2+\left|\mathbf{h}_{\mathrm{rb}}^H \boldsymbol{\Theta}_\mathrm{t} \mathbf{h}_{\mathrm{rc}}^*\right|^2 P_\mathrm{j}+\sigma_\mathrm{b}^2}\Bigg),\label{eq_capac_b}\\
		C_\mathrm{c}=&\log _2 \left(1+\frac{\left|\mathbf{h}_{\mathrm{rc}}^H \boldsymbol{\Theta}_\mathrm{t} \mathbf{H}_{\mathrm{AR}} \mathbf{w}_\mathrm{c}\right|^2}{\left|\mathbf{h}_{\mathrm{rc}}^H \boldsymbol{\Theta}_\mathrm{t} \mathbf{H}_{\mathrm{AR}} \mathbf{w}_\mathrm{b}\right|^2+\left|h_\mathrm{cc}\right|^2 P_\mathrm{j}+\sigma_\mathrm{c}^2}\right).\label{eq_capac_c}
	\end{align}
	Hence, when the required transmission rate between Alice and Bob (or Carol) is selected as $R_\mathrm{b}$ (or $R_\mathrm{c}$), the closed-form expressions of the communication outage probabilities at Bob and Carol can be obtained  through Theorem \ref{th3}.
	\begin{theorem}\label{th3}
		The communication outage probabilities between Alice and Bob/Carol are respectively derived as
		\begin{align}
			\delta_\mathrm{AB} =&\begin{cases}
				0,  &\Upsilon>P_\mathrm{j}^{\max},\\	
				1-\frac{\Upsilon}{ P_\mathrm{j}^{\max}}, &0\leq\Upsilon\leq P_\mathrm{j}^{\max},\\
				1,&\Upsilon<0,
			\end{cases}\label{outage_exp_AB}\\
			\delta_\mathrm{AC}=&\begin{cases}e^{-\frac{\Gamma}{\phi P_\mathrm{j}^{\max}}}+\frac{\Gamma}{\phi P_\mathrm{j}^{\max}}\operatorname{Ei}\left(-\frac{\Gamma}{\phi P_\mathrm{j}^{\max}}\right)	,&\Gamma\geq 0,\\
				1, &\Gamma<0,	
			\end{cases}\label{outage_exp_AC}
		\end{align}
		where $\Upsilon=\frac{\left|\mathbf{h}_{\mathrm{rb}}^H \boldsymbol{\Theta}_\mathrm{r} \mathbf{H}_{\mathrm{AR}} \mathbf{w}_\mathrm{b}\right|^2-\left(2^{R_\mathrm{b}}-1\right)\left(\left|\mathbf{h}_{\mathrm{rb}}^H \boldsymbol{\Theta}_\mathrm{r} \mathbf{H}_{\mathrm{AR}} \mathbf{w}_ \mathrm{c}\right|^2+\sigma_\mathrm{b}^2\right)}{\left(2^{R_\mathrm{b}}-1\right)\left|\mathbf{h}_{\mathrm{rb}}^H \boldsymbol{\Theta}_\mathrm{t} \mathbf{h}_{\mathrm{rc}}^*\right|^2}$ and
		$\Gamma=\frac{\left|\mathbf{h}_{\mathrm{rc}}^H\boldsymbol{\Theta}_\mathrm{t} \mathbf{H}_{\mathrm{AR}} \mathbf{w}_ \mathrm{c}\right|^2-\left(2^{R_\mathrm{c}}-1\right)\left(\left|\mathbf{h}_{\mathrm{rc}}^H\boldsymbol{\Theta}_\mathrm{t} \mathbf{H}_{\mathrm{AR}} \mathbf{w}_ \mathrm{b}\right|^2+\sigma_\mathrm{c}^2\right)}{\left(2^{R_\mathrm{c}}-1\right)}$. Here,  $\operatorname{Ei}(\cdot)$ is the exponential internal function given by
		\begin{equation}\label{eq_expon_fun}
			\operatorname{Ei}(x)=-\int_{-x}^{\infty}\frac{e^{-t}}{t} d t.
		\end{equation}
	\end{theorem}
	
	\begin{proof}
		The proof is given in Appendix \ref{Apd_C}.
	\end{proof}
	
	The communication outage constraints are then defined as $\delta_\mathrm{AB}\leq\iota$ and $\delta_\mathrm{AC}\leq\kappa$ where  $\iota$ and $\kappa$ are two communication outage thresholds required by the system performance indicators for Bob and Carol, respectively.
	In this paper, we try to maximize the covet rate of Bob under the covert constraint $\hat{\mathrm{P}}_\mathrm{e a}^*\geq 1-\epsilon$ and the communication outage constraints.
	It is easy to note that $\delta_\mathrm{AB}$ and $\delta_\mathrm{AC}$ are segment functions with uncertain segment points $\Upsilon$ and $\Gamma$ determined by the active and passive beamforming variables, i.e., $\mathbf{w}_\mathrm{b}$, $\mathbf{w}_\mathrm{c}$  and $\boldsymbol{\Theta}_\mathrm{r}$, $\boldsymbol{\Theta}_\mathrm{t}$, which will be jointly optimized in the next section. Hence, it is difficult to handle the two communication outage constraints in an optimization problem.
	In order to facilitate the optimization and analysis of the considered problem, we equivalently re-express the two communication outage constraints as
	\begin{align}
		\Upsilon\geq P_\mathrm{j}^{\max}(1-\iota) \Longleftrightarrow R_\mathrm{b} \leq R_\mathrm{bb},\label{eq_upper_rate_b}\\
		\Gamma\geq \sigma^* \Longleftrightarrow R_\mathrm{c} \leq R_\mathrm{cc},\label{eq_upper_rate_c}
	\end{align}
	where $\sigma^*$  is the solution to the equation of $\delta_\mathrm{AC}=\kappa$ which can be numerically solved by the bi-section search method.
	In addition, $R_\mathrm{bb}$  and $R_\mathrm{cc}$ are respectively the upper bounds of $R_\mathrm{b}$ and $R_\mathrm{c}$ in order to guarantee the two communication outage constraints which can be expressed as
	\begin{align}
		&\hspace{-11mm}R_\mathrm{bb}=\notag\log _2 \Bigg(1+\\
		&\hspace{-11mm}\frac{\left|\mathbf{h}_{\mathrm{rb}}^H \boldsymbol{\Theta}_\mathrm{r} \mathbf{H}_{\mathrm{AR}} \mathbf{w}_ \mathrm{b}\right|^2}{\left|\mathbf{h}_{\mathrm{rb}}^H \boldsymbol{\Theta}_\mathrm{r} \mathbf{H}_{\mathrm{AR}} \mathbf{w}_ \mathrm{c}\right|^2+\left|\mathbf{h}_{\mathrm{rb}}^H \boldsymbol{\Theta}_\mathrm{t} \mathbf{h}_{\mathrm{rc}}^*\right|^2 P_\mathrm{j}^{\max}(1-\iota)+\sigma_\mathrm{b}^2}\Bigg),\label{eq_outrate_b}\\
		R_\mathrm{cc}=&\log _2 \left(1+\frac{\left|\mathbf{h}_{\mathrm{rc}}^H \boldsymbol{\Theta}_\mathrm{t} \mathbf{H}_{\mathrm{AR}} \mathbf{w}_\mathrm{c}\right|^2}{\left|\mathbf{h}_{\mathrm{rc}}^H \boldsymbol{\Theta}_\mathrm{t} \mathbf{H}_{\mathrm{AR}} \mathbf{w}_\mathrm{b}\right|^2+\sigma^*+\sigma_\mathrm{c}^2}\right). \label{eq_outrate_c}
	\end{align}
	
	Base on the above analysis, we know that the maximum covert rate for Bob under the communication outage constraint is $R_\mathrm{bb}$, then we can maximize $R_\mathrm{bb}$ accordingly to improve the system covert performance. Similarly, the maximum communication rate for Carol under the communication outage constraint is $R_\mathrm{cc}$, and we introduce a constraint $R_\mathrm{cc}\geq R^*$ to guarantee the QoS for the assistant user Carol where $R^*$ is a minimum required communication rate for Carol.

	\section{Problem Formulation and Algorithm Design}\label{sec:S4}
	\subsection{Optimization Problem Formulation}\label{sec:S4_P1}
	On the basis of the previous discussions in section \ref{sec:S3}, we formulate the optimization problem in this section. Specifically, we will maximize the covert rate between Alice and Bob under the covert communication constraint while ensuring the QoS at Carol with the QoS constraint, by jointly optimizing the active and passive beamforming variables, i.e., $\mathbf{w}_\mathrm{b}$, $\mathbf{w}_\mathrm{c}$  and $\boldsymbol{\Theta}_\mathrm{r}$, $\boldsymbol{\Theta}_\mathrm{t}$. Hence, the optimized problem is formulated as
	\begin{subequations}\label{eq_for_opt}
		\begin{align}
			&\hspace{-2mm}\max _{\boldsymbol{\Theta}_\mathrm{r}, \boldsymbol{\Theta}_\mathrm{t}, \mathbf{w}_\mathrm{b}, \mathbf{w}_\mathrm{c}}R_\mathrm{bb},\notag\\
			&\hspace{-2mm}\qquad\text { s.t. }\left\|\mathbf{w}_\mathrm{b}\right\|_2^2+\left\|\mathbf{w}_\mathrm{c}\right\|_2^2 \leq P_{\text {max }}\label{eq_for_opt_1} \\
			&\hspace{-2mm}\qquad\qquad\frac{l_{\mathrm{rw}} \theta_\mathrm{r} \varpi_\mathrm{b}\ln\left( \frac{l_{\mathrm{rw}} \theta_\mathrm{r}\left(\varpi_\mathrm{b}+\varpi_\mathrm{c}\right)}{l_{\mathrm{rw}} \theta_\mathrm{r}\left(\varpi_\mathrm{b}+\varpi_\mathrm{c}\right)+P_\mathrm{j}^{\max} \lambda_{\mathrm{rw}}}\right)}{P_\mathrm{j}^{\max} \lambda_{\mathrm{rw}}} \geq -\epsilon ,\label{eq_for_opt_2}\\
			&\hspace{-2mm}\qquad\qquad R_\mathrm{cc}\geq R^*, \label{eq_for_opt_3}\\
			&\hspace{-2mm}\qquad\qquad\beta_\mathrm{r}^n+\beta_\mathrm{t}^n=1,\label{eq_for_opt_4} \\
			&\hspace{-2mm}\qquad\qquad\phi_\mathrm{r}^n, \phi_\mathrm{t}^n \in[0,2 \pi).\label{eq_for_opt_5}
		\end{align}
	\end{subequations}
	where \eqref{eq_for_opt_1} is the transmission power constraint for Alice with $P_{\max}$ being the maximum transmitted power; \eqref{eq_for_opt_2} is an equivalent covert communication constraint of $\hat{P}_\mathrm{e a}^*\geq 1-\epsilon$; 
	\eqref{eq_for_opt_3} represents the QoS constraint for Carol; 
	\eqref{eq_for_opt_4} and \eqref{eq_for_opt_5} are the amplitude and phase shift constraints for STAR-RIS, respectively.
	Actually, it is challenging to solve the formulated optimization problem because of the following reasons.
	Firstly, the active and passive beamforming variables, $\mathbf{w}_\mathrm{b}$, $\mathbf{w}_\mathrm{c}$  and $\boldsymbol{\Theta}_\mathrm{r}$, $\boldsymbol{\Theta}_\mathrm{t}$ are strongly coupled in the objective function}, covert communication constraint \eqref{eq_for_opt_2} and QoS constraint \eqref{eq_for_opt_3}.
In addition, the utilization of STAR-RIS introduces a characteristic amplitude constraint \eqref{eq_for_opt_4} due to the fact that $\boldsymbol{\Theta}_\mathrm{r}$ and $\boldsymbol{\Theta}_\mathrm{t}$ depend on each other in terms of element amplitudes.

Hence, the traditional convex optimization algorithms cannot be used directly to solve the non-convex optimized problem \eqref{eq_for_opt}. To tackle this issue, the alternating strategy is leveraged to design the optimization algorithm. Specifically, we first divide the original problem into three subproblems where two subproblems are  focused on designing the active beamformer variables $\mathbf{w}_\mathrm{b}$, $\mathbf{w}_\mathrm{c}$, while the passive beamformer variables $\boldsymbol{\Theta}_\mathrm{r}$, $\boldsymbol{\Theta}_\mathrm{t}$ are obtained by solving the last subproblem. After the convergence of the algorithm, we can finally obtain the solution for joint active and passive beamforming.

\vspace{-2mm}
\subsection{Active Beamforming Design}\label{sec:S4_P2}
In this section, we formulate two subproblems based on the original problem \eqref{eq_for_opt} which are solved to design the active beamforming variables $\mathbf{w}_\mathrm{b}$ and $\mathbf{w}_\mathrm{c}$, respectively. 

\subsubsection{Active Beamforming Design for $\mathbf{w}_\mathrm{b}$}
First, we consider the active beamforming design for $\mathbf{w}_\mathrm{b}$ with given $\mathbf{w}_\mathrm{c}$ and the passive beamforming variables $\boldsymbol{\Theta}_\mathrm{r}$, $\boldsymbol{\Theta}_\mathrm{t}$.
In this circumstance, the objective function of the original problem turns into maximizing $\left|\mathbf{h}_{\mathrm{rb}}^H \boldsymbol{\Theta}_\mathrm{r} \mathbf{H}_{\mathrm{AR}} \mathbf{w}_ \mathrm{b}\right|^2$, and the QoS constraint \eqref{eq_for_opt_3} is equivalent to $\left|\mathbf{h}_{\mathrm{rc}}^H \boldsymbol{\Theta}_\mathrm{t} \mathbf{H}_{\mathrm{AR}} \mathbf{w}_\mathrm{b}\right|^2\leq f(R^*)$ with $f(R^*)=\frac{\left|\boldsymbol{\vartheta}_\mathrm{t}^T(\mathbf{h}_{\mathrm{rc}}^*\circ\mathbf{H}_\mathrm{AR})\mathbf{w}_\mathrm{c}\right|^2}{2^{R^*}-1}-\sigma^*-\sigma_\mathrm{c}^2$.
Hence, the corresponding subproblem can be formulated as
\begin{subequations}\label{eq_acti_wb}
\begin{align}
	&\max _{\mathbf{w}_\mathrm{b}}\left|\boldsymbol{\vartheta}_\mathrm{r}^T\mathbf{H}_{\mathrm{rb}}^*\mathbf{H}_{\mathrm{AR}}\mathbf{w}_\mathrm{b}\right|^2 ,\notag \\
	&\text { s.t. }\left\|\mathbf{w}_\mathrm{b}\right\|_2^2 \leq P_{\max }-\left\|\mathbf{w}_\mathrm{c}\right\|_2^2 ,\label{eq_acti_wb_1} \\
	&\qquad \frac{l_{\mathrm{rw}} \theta_\mathrm{r} \varpi_\mathrm{b}\ln \left(1+\frac{P_\mathrm{j}^{\max} \lambda_{\mathrm{rw}}}{l_{\mathrm{rw}} \theta_\mathrm{r}\left(\varpi_\mathrm{b}+\varpi_\mathrm{c}\right)}\right)}{P_\mathrm{j}^{\max} \lambda_{\mathrm{rw}}}\leq \epsilon,\label{eq_acti_wb_2}\\
	&\qquad \left|\boldsymbol{\vartheta}_\mathrm{t}^T\mathbf{H}_{\mathrm{rc}}^*\mathbf{H}_\mathrm{AR}\mathbf{w}_\mathrm{b}\right|^2\leq f(R^*).\label{eq_acti_wb_3}
\end{align}
\end{subequations}
Here, it is easy to verify that 
\begin{align}
\mathbf{h}_{\mathrm{rb}}^H \boldsymbol{\Theta}_\mathrm{r} \mathbf{H}_{\mathrm{AR}}=\boldsymbol{\vartheta}_\mathrm{r}^T\mathbf{H}_{\mathrm{rb}}^*\mathbf{H}_{\mathrm{AR}}, \label{eq_trans_r} \\
\mathbf{h}_{\mathrm{rc}}^H \boldsymbol{\Theta}_\mathrm{t} \mathbf{H}_{\mathrm{AR}}=\boldsymbol{\vartheta}_\mathrm{t}^T\mathbf{H}_\mathrm{rc}^*\mathbf{H}_\mathrm{AR},\label{eq_trans_t}
\end{align}
by vertorizing the diagonal matrixes $\boldsymbol{\Theta}_\mathrm{r}$ and $\boldsymbol{\Theta}_\mathrm{t}$ as $\boldsymbol{\vartheta}_\mathrm{r}=\operatorname{diag}(\boldsymbol{\Theta}_\mathrm{r})$ and $\boldsymbol{\vartheta}_\mathrm{t}=\operatorname{diag}(\boldsymbol{\Theta}_\mathrm{t})$, and matrixing the vectors $\mathbf{h}_{\mathrm{rb}}$ and $\mathbf{h}_{\mathrm{rc}}$ as $\mathbf{H}_{\mathrm{rb}}=\operatorname{Diag}(\mathbf{h}_{\mathrm{rb}})$ and $\mathbf{H}_{\mathrm{rc}}=\operatorname{Diag}(\mathbf{h}_{\mathrm{rc}})$.    
It is worth noting that the optimization problem \eqref{eq_acti_wb} is still non-convex due to the non-concave objective function and the non-convex covert communication constraint \eqref{eq_acti_wb_2} recalling that $\varpi_\mathrm{b}=\mathbf{w}_\mathrm{b}^H\mathbf{w}_\mathrm{b}$.

To effectively address the non-convex problem \eqref{eq_acti_wb}, we resort to the SDR method \cite{luo10}.  
In particular, let $\mathbf{W}_\mathrm{b}=\mathbf{w}_\mathrm{b}\mathbf{w}_\mathrm{b}^H$, then problem \eqref{eq_acti_wb} can be equivalently transformed as
\begin{subequations}\label{eq_trans_acti_wb}
\begin{align}
	&\max _{\mathbf{W}_\mathrm{b}}\operatorname{Tr}(\mathbf{A}\mathbf{W}_\mathrm{b}) ,\notag \\
	&\text { s.t. }\operatorname{Tr}(\mathbf{W}_\mathrm{b}) \leq P_{\max }-\left\|\mathbf{w}_\mathrm{c}\right\|_2^2 ,\label{eq_trans_acti_wb_1} \\
	&\qquad \varpi_\mathrm{b}\ln \left(1+\frac{P_\mathrm{j}^{\max} \lambda_{\mathrm{rw}}}{l_{\mathrm{rw}} \theta_\mathrm{r}\left(\varpi_\mathrm{b}+\varpi_\mathrm{c}\right)}\right)\leq \frac{\epsilon P_\mathrm{j}^{\max} \lambda_{\mathrm{rw}}}{l_{\mathrm{rw}} \theta_\mathrm{r} },\label{eq_trans_acti_wb_2}\\
	&\qquad \operatorname{Tr}(\mathbf{B}\mathbf{W}_\mathrm{b})\leq f(R^*),\label{eq_trans_acti_wb_3}\\
	&\qquad \mathbf{W}_\mathrm{b}\succeq \mathbf{0},\label{eq_trans_acti_wb_4}\\
	&\qquad \operatorname{rank}(\mathbf{W}_\mathrm{b})=1,\label{eq_trans_acti_wb_5}
\end{align}
\end{subequations}
where $\mathbf{A}=(\mathbf{H}_\mathrm{rb}^*\mathbf{H}_\mathrm{AR})^H\boldsymbol{\vartheta}_\mathrm{r}^*\boldsymbol{\vartheta}_\mathrm{r}^T(\mathbf{H}_\mathrm{rb}^*\mathbf{H}_\mathrm{AR}), \mathbf{B}= (\mathbf{H}_{\mathrm{rc}}^*\mathbf{H}_\mathrm{AR})^H$ $\boldsymbol{\vartheta}_\mathrm{t}^*\boldsymbol{\vartheta}_\mathrm{t}^T(\mathbf{H}_{\mathrm{rc}}^*\mathbf{H}_\mathrm{AR})$.
Although \eqref{eq_trans_acti_wb_2} is still a  non-convex constraint about $\mathbf{W}_\mathrm{b}$, we note that the left-side of \eqref{eq_trans_acti_wb_2}, denoted as $g(\varpi_\mathrm{b})$, is a concave function of $\varpi_\mathrm{b}=\operatorname{Tr}(\mathbf{W}_\mathrm{b})$.
Thus, the first-order Taylor expansion of $g(\varpi_\mathrm{b})$ can be leveraged to replace it iteratively, which is an upper-bound linear approximation and generate a tighter convex substitute for \eqref{eq_trans_acti_wb_2}. 
Specifically,  in the $(m+1)$-th iteration of the overall proposed algorithm ($m=0,1,\cdots$), the first-order Taylor expansion of  $g(\varpi_\mathrm{b})$, with the given point $\varpi_\mathrm{b}^{(m)}$ obtained in the $m$-th iteration,  can be expressed as 
\begin{equation}\label{eq_tylorexpan_wb}
\begin{aligned}
	\hspace{-2mm}\widehat{g}(\varpi_\mathrm{b},\varpi_\mathrm{b}^{(m)})=g(\varpi_\mathrm{b}^{(m)})+ \frac{\partial g}{\partial\mathrm{\varpi_\mathrm{b}}}(\varpi_\mathrm{b}^{(m)})*(\varpi_\mathrm{b}-\varpi_\mathrm{b}^{(m)}),
\end{aligned}
\end{equation}
where $ \frac{\partial g}{\partial\mathrm{\varpi_\mathrm{b}}}(\varpi_\mathrm{b}^{(m)})$ is given by \eqref{eq_trans_acti_wb_partial} at the top of the next page.
\begin{figure*}[t]
\hrulefill
\vspace*{4pt}
\begin{equation}\label{eq_trans_acti_wb_partial}
	\begin{aligned}
		\frac{\partial g}{\partial\mathrm{\varpi_\mathrm{b}}}(\varpi_\mathrm{b}^{(m)})=\ln \left(1+\frac{P_\mathrm{j}^{\max} \lambda_{\mathrm{rw}}}{l_{\mathrm{rw}} \theta_\mathrm{r}\left(\varpi_\mathrm{b}^{(m)}+\varpi_\mathrm{c}\right)}\right)+\frac{-\varpi_\mathrm{b}^{(m)}
			l_\mathrm{rw}\theta_\mathrm{r}P_\mathrm{j}^{\max}\lambda_\mathrm{rw}}{(\varpi_\mathrm{b}^{(m)}+\varpi_\mathrm{c})\left(l_\mathrm{rw}\theta_\mathrm{r}(\varpi_\mathrm{b}^{(m)}
			+\varpi_\mathrm{c})+P_\mathrm{j}^{\max}\lambda_\mathrm{rw}\right)},
	\end{aligned}
\end{equation}
\end{figure*}

\vspace{2mm}
Therefore, the problem \eqref{eq_trans_acti_wb} in the $(m+1)$-th iteration can be reformulated as
\begin{subequations}\label{eq_trans_acti_wb_remove}
\begin{align}
	&\max _{\mathbf{W}_\mathrm{b}}\operatorname{Tr}(\mathbf{A}\mathbf{W}_\mathrm{b}) ,\notag \\
	&\text { s.t. }	\widehat{g}(\varpi_\mathrm{b},\varpi_\mathrm{b}^{(m)})\leq\frac{\epsilon P_\mathrm{j}^{\max} \lambda_{\mathrm{rw}}}{l_{\mathrm{rw}} \theta_\mathrm{r} } ,\label{eq_trans_acti_wb_remove_1} \\
	&\qquad \eqref{eq_trans_acti_wb_1}, \eqref{eq_trans_acti_wb_3}, \eqref{eq_trans_acti_wb_4}, \eqref{eq_trans_acti_wb_5} .\label{eq_trans_acti_wb_remove_2}
\end{align}
\end{subequations}
Note that the remaining non-convexity of problem \eqref{eq_trans_acti_wb_remove} is the non-convex rank-one constraint \eqref{eq_trans_acti_wb_5}. To tackle this issue, we employ the SDR techniques to remove this constraint, and then problem \eqref{eq_trans_acti_wb_remove} can be transformed into a standard convex semidefinite programming (SDP) problem which is able to be effectively solved by existing convex optimization solvers such as CVX \cite{grant14cvx}. As stated in \cite{luo10,wu19}, there is a relatively high probability that the obtained optimal SDP solution cannot satisfy the rank-one constraint, i.e., $\operatorname{rank}(\mathbf{W}_\mathrm{b})\neq 1$. Thus, it is necessary to implement additional steps to construct the rank-one solution from the obtained higher-rank solution, through the commonly used eigenvector approximation or Gaussian randomization methods \cite{luo10,wu18}.

\subsubsection{Active Beamforming Design for $\mathbf{w}_\mathrm{c}$}
On the other hand, for given $\mathbf{w}_\mathrm{b}$, $\boldsymbol{\Theta}_\mathrm{r}$ and $\boldsymbol{\Theta}_\mathrm{t}$, the original optimization problem \eqref{eq_for_opt} can be simplified as
\begin{subequations}\label{eq_activ_wc}
\begin{align}
	&\min _{\mathbf{w}_\mathrm{c}}\left|\boldsymbol{\vartheta}_\mathrm{r}^T\mathbf{H}_{\mathrm{rb}}^*\mathbf{H}_{\mathrm{AR}}\mathbf{w}_\mathrm{c}\right|^2 ,\notag \\
	&\text { s.t. }\left\|\mathbf{w}_\mathrm{c}\right\|_2^2 \leq P_{\max }-\left\|\mathbf{w}_\mathrm{b}\right\|_2^2 ,\label{eq_activ_wc_1} \\
	&\qquad \varpi_\mathrm{c}\leq C(\epsilon),\label{eq_activ_wc_2}\\
	&\qquad \left|\boldsymbol{\vartheta}_\mathrm{t}^T\mathbf{H}_{\mathrm{rc}}^*\mathbf{H}_\mathrm{AR}\mathbf{w}_\mathrm{c}\right|^2\geq \hat{f}(R^*), \label{eq_activ_wc_3}
\end{align}
\end{subequations}
where $C(\epsilon)=\frac{P_\mathrm{j}^{\max}\lambda_\mathrm{rw}}{l_\mathrm{rw}\theta_\mathrm{r}\Big(e^\frac{\epsilon P_\mathrm{j}^{\max}\lambda_\mathrm{rw}}{l_\mathrm{rw}\theta_\mathrm{r}\varpi_\mathrm{b}}-1\Big)}-\varpi_\mathrm{b}$,
$\hat{f}(R^*)=(2^{R^*}-1)\Big(\left|\boldsymbol{\vartheta}_\mathrm{t}^T\mathbf{H}_{\mathrm{rc}}^*\mathbf{H}_\mathrm{AR}\mathbf{w}_\mathrm{b}\right|^2+
\sigma^*+\sigma_\mathrm{c}^2\Big)$, and $\varpi_\mathrm{c}=\mathbf{w}_\mathrm{c}^H\mathbf{w}_\mathrm{c}$.

Similarly, we leverage the SDR method to handle the non-convex optimized problem \eqref{eq_activ_wc}. By defining $\mathbf{W}_\mathrm{c}=\mathbf{w}_\mathrm{c}\mathbf{w}_\mathrm{c}^H$, the problem can be  transformed  as
\begin{subequations}\label{eq_activ_wc_trans}
\begin{align}
	&\min _{\mathbf{W}_\mathrm{c}}\operatorname{Tr}(\mathbf{A}\mathbf{W}_\mathrm{c}) ,\notag \\
	&\text { s.t. }\operatorname{Tr}(\mathbf{W}_\mathrm{c}) \leq P_{\max }-\left\|\mathbf{w}_\mathrm{b}\right\|_2^2 ,\label{eq_activ_wc_trans_1} \\
	&\qquad \operatorname{Tr}(\mathbf{W}_\mathrm{c})\leq C(\epsilon),\label{eq_activ_wc_trans_2}\\
	&\qquad \operatorname{Tr}(\mathbf{B}\mathbf{W}_\mathrm{c})\geq\hat{f}(R^*), \label{eq_activ_wc_trans_3}\\
	&\qquad \operatorname{rank}(\mathbf{W}_\mathrm{c})=1,\label{eq_activ_wc_trans_4}\\
	&\qquad \mathbf{W}_\mathrm{c}\succeq \mathbf{0}.\label{eq_activ_wc_trans5}
\end{align}
\end{subequations}
The problem \eqref{eq_activ_wc_trans} can be optimally solved by removing the rank-one constraint \eqref{eq_activ_wc_trans_4}, and then the rank-one solution can be constructed through the existing approaches. 

\vspace{-4mm}
\subsection{Joint Passive Beamforming Design for STAR-RIS}\label{sec:S4_P3}
After the active beamforming design, we can optimize the passive beamforming variables $\boldsymbol{\vartheta}_\mathrm{r}$ and $\boldsymbol{\vartheta}_\mathrm{t}$, with  $\mathbf{w}_\mathrm{b}$ and $\mathbf{w}_\mathrm{c}$ being fixed as the obtained solutions in the previous subsection. On the basis of  the original problem \eqref{eq_for_opt}, the corresponding optimization problem for joint passive beamforming for STAR-RIS design can be expressed as 
\begin{subequations}\label{eq_passi}
\begin{align}
	&\max _{\boldsymbol{\vartheta}_\mathrm{r}, \boldsymbol{\vartheta}_\mathrm{t}}\gamma_{bb}(\boldsymbol{\vartheta}_\mathrm{r}, \boldsymbol{\vartheta}_\mathrm{t}),\notag \\%
	&\text { s.t. }~\frac{l_{\mathrm{rw}} \theta_\mathrm{r} \varpi_\mathrm{b}\ln \left(1+\frac{P_\mathrm{j}^{\max} \lambda_{\mathrm{rw}}}{l_{\mathrm{rw}} \theta_\mathrm{r}\left(\varpi_\mathrm{b}+\varpi_\mathrm{c}\right)}\right)}{P_\mathrm{j}^{\max} \lambda_{\mathrm{rw}}}\leq \epsilon,\label{eq_passi_1} \\
	&\qquad~\tilde{f}(\boldsymbol{\vartheta}_\mathrm{t})\geq0,\label{eq_passi_2}\\
	&\qquad~\beta_\mathrm{r}^n+\beta_\mathrm{t}^n=1,\label{eq_passi_3} \\
	&\qquad~\phi_\mathrm{r}^n, \phi_\mathrm{t}^n \in[0,2 \pi),\label{eq_passi_4}
\end{align}
\end{subequations}
where 
\begin{align}
\gamma_{bb}=&\frac{\left|\mathbf{h}_{\mathrm{rb}}^H \boldsymbol{\Theta}_\mathrm{r} \mathbf{H}_{\mathrm{AR}} \mathbf{w}_ \mathrm{b}\right|^2}{\left|\mathbf{h}_{\mathrm{rb}}^H \boldsymbol{\Theta}_\mathrm{r} \mathbf{H}_{\mathrm{AR}} \mathbf{w}_ \mathrm{c}\right|^2+\left|\mathbf{h}_{\mathrm{rb}}^H \boldsymbol{\Theta}_\mathrm{t} \mathbf{h}_{\mathrm{rc}}^*\right|^2 P_\mathrm{j}^{\max}(1-\iota)+\sigma_\mathrm{b}^2},\notag\\
=&\frac{\boldsymbol{\vartheta}_\mathrm{r}^T\mathbf{C}\boldsymbol{\vartheta}_\mathrm{r}^*}{\boldsymbol{\vartheta}_\mathrm{r}^T\mathbf{D}\boldsymbol{\vartheta}_\mathrm{r}^*+
	\boldsymbol{\vartheta}_\mathrm{t}^T\mathbf{E}\boldsymbol{\vartheta}_\mathrm{t}^* P_\mathrm{j}^{\max}(1-\iota)+\sigma_\mathrm{b}^2},\label{eq_passi_obj_trans} \\
&\mathbf{C}=(\mathbf{H}_{\mathrm{rb}}^*\mathbf{H}_{\mathrm{AR}})\mathbf{w}_\mathrm{b}\mathbf{w}_\mathrm{b}^H(\mathbf{H}_{\mathrm{rb}}^*\mathbf{H}_{\mathrm{AR}})^H,\label{eq_passi_C}\\
&\mathbf{D}=(\mathbf{H}_{\mathrm{rb}}^*\mathbf{H}_{\mathrm{AR}})\mathbf{w}_\mathrm{c}\mathbf{w}_\mathrm{c}^H(\mathbf{H}_{\mathrm{rb}}^*\mathbf{H}_{\mathrm{AR}})^H,\label{eq_passi_D} \\
&\mathbf{E}=(\mathbf{H}_\mathrm{rb}^*\mathbf{H}_{\mathrm{rc}}^*)(\mathbf{H}_\mathrm{rb}^*\mathbf{H}_{\mathrm{rc}}^*)^H,\label{eq_passi_E}
\end{align}
\begin{align}
\tilde{f}(\boldsymbol{\vartheta}_\mathrm{t})=& \left|\boldsymbol{\vartheta}_\mathrm{t}^T\mathbf{H}_{\mathrm{rc}}^*\mathbf{H}_\mathrm{AR}\mathbf{w}_\mathrm{c}\right|^2-(2^{R^*}-1)\times,\notag\\
&\left(\left|\boldsymbol{\vartheta}_\mathrm{t}^T\mathbf{H}_{\mathrm{rc}}^*\mathbf{H}_\mathrm{AR}\mathbf{w}_\mathrm{b}\right|^2+\sigma^*+\sigma_\mathrm{c}^2\right).\label{eq_eq_passi_2}
\end{align}
We can find that the fractional objective function and the constraints \eqref{eq_passi_1}, \eqref{eq_passi_2} in problem \eqref{eq_passi} are non-convex w.r.t. $\boldsymbol{\vartheta}_\mathrm{r}$ and $\boldsymbol{\vartheta}_\mathrm{t}$ recalling that $\theta_\mathrm{r}=\operatorname{diag}(\boldsymbol{\Theta}_\mathrm{r})^H\operatorname{diag}(\boldsymbol{\Theta}_\mathrm{r})=\boldsymbol{\vartheta}_\mathrm{r}^H\boldsymbol{\vartheta}_\mathrm{r}$ and $\lambda_\mathrm{rw}=\left\|\boldsymbol{\Theta}_\mathrm{t}\mathbf{h}_{\mathrm{rc}}^*\right\|^2_2=\left\|\boldsymbol{\vartheta}_\mathrm{t}\circ\mathbf{h}_{\mathrm{rc}}^*\right\|^2_2$, which makes this  problem difficult to be solved directly.

In order to effectively solve the optimization problem \eqref{eq_passi}, we first adopt the Dinkelbach's algorithm \cite[Chapter3.2.1]{zap15} and the SDR techniques to deal with the objective function. In the $(i+1)$-th iteration of the Dinkelbach's algorithm, the objective function can be transformed as
\vspace{-2mm}
\begin{equation}\label{eq_Dink_obj_trans}
\begin{aligned}[b]	\widehat{\gamma}_\mathrm{bb}^{(i+1)}(\mathbf{Q}_\mathrm{r},\mathbf{Q}_\mathrm{t})=&\operatorname{Tr}(\mathbf{C}\mathbf{Q}_\mathrm{r})-\chi^{(i)}\Big(\operatorname{Tr}(\mathbf{D}\mathbf{Q}_\mathrm{r})+\\
	&~ \operatorname{Tr}(\mathbf{E}\mathbf{Q}_\mathrm{t})P_\mathrm{j}^{\max}(1-\iota)+\sigma_\mathrm{b}^2\Big),
\end{aligned}
\end{equation}
where $\mathbf{Q}_\mathrm{r}=\boldsymbol{\vartheta}_\mathrm{r}^*\boldsymbol{\vartheta}_\mathrm{r}^T$ and $\mathbf{Q}_\mathrm{t}=\boldsymbol{\vartheta}_\mathrm{t}^*\boldsymbol{\vartheta}_\mathrm{t}^T$ will be optimized to obtain the reflection and transmission coefficients of STAR-RIS.
In addition, the parameter $\chi^{(i)}$ is updated with
\begin{align}\label{eq_Dink_coeff}
\chi^{(i)}=\frac{\operatorname{Tr}(\mathbf{C}\mathbf{Q}_\mathrm{r}^{(i)})}{\left(\operatorname{Tr}(\mathbf{D}\mathbf{Q}_\mathrm{r}^{(i)})
+\operatorname{Tr}(\mathbf{E}\mathbf{Q}_\mathrm{t}^{(i)})P_\mathrm{j}^{\max}(1-\iota)+\sigma_\mathrm{b}^2\right)}, \end{align}
where $\mathbf{Q}_\mathrm{r}^{(i)}$ and $\mathbf{Q}_\mathrm{t}^{(i)}$ are the optimized solution through the $i$-th iteration. In this way,  the objective function has been transformed into an affine function w.r.t. $\mathbf{Q}_\mathrm{r}$ and $\mathbf{Q}_\mathrm{t}$.

Next, we try to deal with the constraints in problem \eqref{eq_passi}. For the left-side of covert communication constraint \eqref{eq_passi_1}, we find that it is a monotonically decreasing function of $\frac{\lambda_\mathrm{rw}}{\theta_\mathrm{r}}$. Thus, the constraint \eqref{eq_passi_1} can be equivalently rewritten as
\begin{equation}\label{eq_trans_eq_passi_q}
\begin{aligned}[b]
\frac{\lambda_\mathrm{rw}}{\theta_\mathrm{r}}\geq\varphi(\epsilon),
\end{aligned}
\end{equation}
where $\varphi(\epsilon)$ 
can be numerically obtained by employing the bisection search method.  
According to the expressions of $\lambda_\mathrm{rw}$ and $\theta_\mathrm{r}$, they can be reformulated as
\begin{equation}\label{eq_reform_theta_lambda_rw}
\begin{aligned}[b]
\lambda_\mathrm{rw}=\boldsymbol{\beta}_\mathrm{t}^T(\mathbf{h}_\mathrm{rc}\circ\mathbf{h}_{\mathrm{rc}}^*), ~ \theta_\mathrm{r}=\boldsymbol{\beta}_\mathrm{r}^T\mathbf{I}_{N\times1}.
\end{aligned}
\end{equation}
where $\boldsymbol{\beta}_\mathrm{r}=\big[\beta_\mathrm{r}^1,\cdots, \beta_\mathrm{r}^N\big]^T$,
$\boldsymbol{\beta}_\mathrm{t}=\big[\beta_\mathrm{t}^1,\cdots, \beta_\mathrm{t}^N\big]^T$.

Similarly, by using SDR techniques, the equivalent form of $\tilde{f}(\boldsymbol{\vartheta}_\mathrm{t})$ in constraint \eqref{eq_passi_2} w.r.t. $\mathbf{Q}_\mathrm{t}$ can be derived as
\begin{equation}\label{eq_trans_eq_passi_2}
\begin{aligned}[b]
\hspace{-3mm}\tilde{f}(\mathbf{Q}_\mathrm{t})=\operatorname{Tr}(\mathbf{F}\mathbf{Q}_\mathrm{t})-\big(2^{R^*}-1\big)
\left(\operatorname{Tr}(\mathbf{G}\mathbf{Q}_\mathrm{t})+\sigma^*+\sigma_\mathrm{c}^2\right),
\end{aligned}
\end{equation}
where we have
\begin{subequations}\label{eq_passi_FG}
\begin{align}
&\qquad 	\mathbf{F}=(\mathbf{H}_{\mathrm{rc}}^*\mathbf{H}_{\mathrm{AR}})\mathbf{w}_\mathrm{c}\mathbf{w}_\mathrm{c}^H(\mathbf{H}_{\mathrm{rc}}^*\mathbf{H}_{\mathrm{AR}})^H,\label{eq_passi_FG_1}\\
&\qquad\mathbf{G}=(\mathbf{H}_{\mathrm{rc}}^*\mathbf{H}_{\mathrm{AR}})\mathbf{w}_\mathrm{b}\mathbf{w}_\mathrm{b}^H(\mathbf{H}_{\mathrm{rc}}^*\mathbf{H}_{\mathrm{AR}})^H.\label{eq_passi_FG_2}
\end{align}
\end{subequations}

Based on the above analysis, the optimization problem \eqref{eq_passi} in the $(i+1)$-th iteration can be formulated as
\begin{subequations}\label{eq_passi_trans_rank1}
\begin{align}
&\hspace{-5mm}\max _{\mathbf{Q}_\mathrm{r}, \mathbf{Q}_\mathrm{t}, \boldsymbol{\beta}_\mathrm{r}, \boldsymbol{\beta}_\mathrm{t}} \widehat{\gamma}_\mathrm{bb}^{(i+1)}(\mathbf{Q}_\mathrm{r},\mathbf{Q}_\mathrm{t}),\notag \\
&\text { s.t. }~\frac{\boldsymbol{\beta}_\mathrm{t}^T(\mathbf{h}_\mathrm{rc}\circ\mathbf{h}_{\mathrm{rc}}^*)}{\boldsymbol{\beta}_\mathrm{r}^T\mathbf{I}_{N\times1}}\geq \varphi(\epsilon),\label{eq_passi_trans_rank1_1} \\
&\qquad\tilde{f}(\mathbf{Q}_\mathrm{t})\geq0, \label{eq_passi_trans_rank1_2}\\
&\qquad\beta_\mathrm{r}^n+\beta_\mathrm{t}^n=1,\label{eq_passi_trans_rank1_3} \\
&\qquad\operatorname{diag}(\mathbf{Q}_\mathrm{r})=\boldsymbol{\beta}_\mathrm{r}, \operatorname{diag}(\mathbf{Q}_\mathrm{t})=\boldsymbol{\beta}_\mathrm{t},\label{eq_passi_trans_rank1_4}\\
&\qquad\operatorname{rank}(\mathbf{Q}_\mathrm{r})=1, \operatorname{rank}(\mathbf{Q}_\mathrm{t})=1,\label{eq_passi_trans_rank1_5}\\
&\qquad\mathbf{Q}_\mathrm{r}\succeq0, \mathbf{Q}_\mathrm{t}\succeq0.\label{eq_passi_trans_rank1_6}
\end{align}
\end{subequations}

However, problem \eqref{eq_passi_trans_rank1} is still a non-convex optimization problem due to the two rank-one constraints in \eqref{eq_passi_trans_rank1_5}.
Due to the dependence of $\mathbf{Q}_\mathrm{r}$ and  $\mathbf{Q}_\mathrm{t}$, it is difficult to re-construct the rank-one solution if we remove the rank-one constraints directly.  To handle this issue, we equivalently rewrite the rank-one constraints as  \cite{X.HU_TCOM21RIS}
\begin{subequations}\label{eq_trans_rank1}
\begin{align}
\operatorname{rank}(\mathbf{Q}_\mathrm{r})=&1 \Longleftrightarrow \eta_\mathrm{r}\triangleq\operatorname{Tr}(\mathbf{Q}_\mathrm{r})-\|\mathbf{Q}_\mathrm{r}\|_2=0, \label{eq_trans_rank1_1}\\
\operatorname{rank}(\mathbf{Q}_\mathrm{t})=&1 \Longleftrightarrow \eta_\mathrm{t}\triangleq\operatorname{Tr}(\mathbf{Q}_\mathrm{t})-\|\mathbf{Q}_\mathrm{t}\|_2=0,\label{eq_trans_rank1_2}
\end{align}
\end{subequations}
where $\|\mathbf{Q}\|_2$ represents the spectral norm  which is a convex function of $\mathbf{Q}$.
Note that for any positive semidefinite matrix $\mathbf{Q}\succeq0$, the inequality $\operatorname{Tr}(\mathbf{Q})-\|\mathbf{Q}\|_2\geq0$ always holds and the equality is satisfied if and only if $\operatorname{rank}(\mathbf{Q})=1$. Based on the non-negative feature of $\eta_\mathrm{r}$ and $\eta_\mathrm{t}$, we add them  into the objective function of problem \eqref{eq_passi_trans_rank1} as penalty terms for the rank-one constraints. Hence, the optimization problem  can be re-expressed as
\begin{subequations}\label{eq_passi_trans_withoutrank1}
\begin{align}
&\max _{\mathbf{Q}_\mathrm{r}, \mathbf{Q}_\mathrm{t}, \boldsymbol{\beta}_\mathrm{r}, \boldsymbol{\beta}_\mathrm{t}} \widehat{\gamma}_\mathrm{bb}^{(i+1)}(\mathbf{Q}_\mathrm{r},\mathbf{Q}_\mathrm{t})-\rho_1\eta_\mathrm{r}-\rho_2\eta_\mathrm{t},\notag \\
&\qquad\text { s.t. }~ \frac{\boldsymbol{\beta}_\mathrm{t}^T(\mathbf{h}_\mathrm{rc}\circ\mathbf{h}_{\mathrm{rc}}^*)}{\boldsymbol{\beta}_\mathrm{r}^T\mathbf{I}_{N\times1}}\geq \varphi(\epsilon),\label{eq_passi_trans_withoutrank1_1} \\
&\qquad\qquad \eqref{eq_passi_trans_rank1_2}, \eqref{eq_passi_trans_rank1_3},\eqref{eq_passi_trans_rank1_4},\eqref{eq_passi_trans_rank1_6},\label{eq_passi_trans_withoutrank1_3}
\end{align}
\end{subequations}
where $\rho_1$, $\rho_2>$0 are the introduced penalty coefficients. Now, the optimization problem \eqref{eq_passi_trans_withoutrank1} is still non-convex  because of the non-convexity of the penalty terms $\eta_\mathrm{r}$ and $\eta_\mathrm{t}$.
By replacing the convex spectral norms  in $\eta_\mathrm{r}$ and $\eta_\mathrm{t}$ with their linear lower-bound, i.e., first-order Taylor expansions, we can obtain the upper-bound linear approximations for $\eta_\mathrm{r}$ and $\eta_\mathrm{t}$ as 
\begin{align}
\eta_\mathrm{r}\leq &\operatorname{Tr}(\mathbf{Q}_\mathrm{r})-\left(\|\mathbf{Q}_\mathrm{r}^{(i)}\|_2
+\operatorname{Tr}\left(\mathbf{q}_\mathrm{r}^{(i)}(\mathbf{q}_\mathrm{r}^{(i)})^H(\mathbf{Q}_\mathrm{r}-\mathbf{Q}_\mathrm{r}^{(i)})\right)\right) \nonumber\\
=&\widehat{\eta}_\mathrm{r}(\mathbf{Q}_\mathrm{r}),\label{eq_penal_trans_r}\\
\eta_\mathrm{t}\leq &\operatorname{Tr}(\mathbf{Q}_\mathrm{t})-\left(\|\mathbf{Q}_\mathrm{t}^{(i)}\|_2
+\operatorname{Tr}\left(\mathbf{q}_\mathrm{t}^{(i)}(\mathbf{q}_\mathrm{t}^{(i)})^H(\mathbf{Q}_\mathrm{t}-\mathbf{Q}_\mathrm{t}^{(i)})\right)\right) \nonumber\\
=&\widehat{\eta}_\mathrm{t}(\mathbf{Q}_\mathrm{t}),\label{eq_penal_trans_t}
\end{align}
where $\mathbf{q}_\mathrm{r}^{(i)}$ and $\mathbf{q}_\mathrm{t}^{(i)}$ are the eigenvectors corresponding to the largest eigenvalues of $\mathbf{Q}_\mathrm{r}^{(i)}$ and $\mathbf{Q}_\mathrm{t}^{(i)}$. 
Hence, $\widehat{\gamma}_\mathrm{bb}^{(i+1)}(\mathbf{Q}_\mathrm{r},\mathbf{Q}_\mathrm{t})-\rho_1\widehat{\eta}_\mathrm{r}(\mathbf{Q}_\mathrm{r})
-\rho_2\widehat{\eta}_\mathrm{t}(\mathbf{Q}_\mathrm{t})$ is a linear lower-bound of the objective function in problem \eqref{eq_passi_trans_withoutrank1}, which will be further utilized as the objective function in  problem \eqref{eq_passi_trans_withoutrank1} to obtain the solution of $\mathbf{Q}_\mathrm{r}$ and $\mathbf{Q}_\mathrm{t}$ in the $(i+1)$-th iteration, denoted $\mathbf{Q}_\mathrm{r}^{(i+1)}$ and $\mathbf{Q}_\mathrm{t}^{(i+1)}$.

In conclusion, to solve the optimization problem \eqref{eq_passi} for joint passive beamforming design, we propose a two-tier iterative algorithm as summarized in Algorithm 1, 
where the outer loop is for updating the penalty coefficients $\rho_1$ and $\rho_2$ through the penalty method \cite{wright99} and the inner loop is for updating $\mathbf{Q}_\mathrm{r}$ and $\mathbf{Q}_\mathrm{t}$ through the Dinkelbach's algorithm \cite{zap15,J_X.Hu19Edge}. 
Here, $\omega>1$ is the scaling factor of the penalty coefficient. Also, $v_1>0$ denotes the penalty violation and  $v_2>0$ indicates the  gap of the objective functions between two adjacent iterations of the Dinkelbach's algorithm, which are given by
\begin{align}
v_1=&\max\{ \eta_\mathrm{r},\eta_\mathrm{t}\},\label{eq_violation_1}\\	v_2=&\left|\widehat{\gamma}_\mathrm{bb}^{(i+1)}(\mathbf{Q}_\mathrm{r}^{(i+1)},\mathbf{Q}_\mathrm{t}^{(i+1)})-\widehat{\gamma}_\mathrm{bb}^{(i)}(\mathbf{Q}_\mathrm{r}^{(i)},\mathbf{Q}_\mathrm{t}^{(i)})\right|.\label{eq_violation_2}
\end{align}
\begin{center}
\begin{tabular}{p{8.5cm}}
\toprule[2pt]
\textbf{Algorithm 1:} Proposed Iterative Algorithm for Problem \eqref{eq_passi} on Joint Passive Beamforming Design of STAR-RIS  \\ 
\midrule[1pt]
1: Initialize feasible point ($\mathbf{Q}_\mathrm{r}^{(0)}$, $\mathbf{Q}_\mathrm{t}^{(0)}$), penalty coefficients \\
~~($\rho_1^{(0)}$, $\rho_2^{(0)}$), and calculate $v_1$; Define the accuracy tolerance\\
~~thresholds $\varepsilon_1$, $\varepsilon_2$; Set iteration index $l=0$ for outer loop.\\
2: \textbf{While} $v_1>\varepsilon_1$  or $l=0$ \textbf{do}                       \\
3: \quad Initialize $\chi^{(0)}$ and set $i=0$ for  inner loop.\\
4: \quad\textbf{While} $v_2>\varepsilon_2 $ or $i=0$ \textbf{do}\\
5: \qquad Solve the problem \eqref{eq_passi_trans_withoutrank1} with given ($\mathbf{Q}_\mathrm{r}^{(i)}$, $\mathbf{Q}_\mathrm{t}^{(i)}$).\\
6: \qquad Update the ($\mathbf{Q}_\mathrm{r}^{(i+1)}$, $\mathbf{Q}_\mathrm{t}^{(i+1)}$) with  obtained solution.\\
7:  \qquad Calculate $v_2$, $\chi^{(i+1)}$ based on the obtained solution,\\ \quad\qquad and let $i=i+1$.\\
8: \quad\textbf{end while}\\
9: \quad Calculate $v_1$; Update $\rho_1^{(l+1)}=\omega\rho_1^{(l)}$, $\rho_2^{(l+1)}=\omega\rho_2^{(l)}$;\\
~~~~~~Let ($\mathbf{Q}_\mathrm{r}^{(0)}$, $\mathbf{Q}_\mathrm{t}^{(0)}$)$=$($\mathbf{Q}_\mathrm{r}^{(i)}$, $\mathbf{Q}_\mathrm{t}^{(i)}$) and $l=l+1$.\\
10: \textbf{end while}      \\
11: Calculate the $\boldsymbol{\Theta}_\mathrm{r}$, $\boldsymbol{\Theta}_\mathrm{t}$ through the obtained $\mathbf{Q}_\mathrm{r}$, $\mathbf{Q}_\mathrm{t}$.      \\
\bottomrule[2pt]
\end{tabular}
\end{center}

\subsection{Proposed Optimization Algorithm \& Analysis on Complexity and Convergence}\label{sec:S4_P4}
Algorithm 2 concludes the overall processes for solving the original optimization problem \eqref{eq_for_opt} for STAR-RIS-assisted covert communications, which is an alternating optimization algorithm to solve three subproblems alternatively as detailed in Section \ref{sec:S4}. Here, $v$$>$$0$ represents the gap of objective values between two adjacent iterations, 
and the algorithm converges  when $v$ is below a predefined accuracy threshold $\varepsilon$.

Next, we give the analysis on the computational complexity of the proposed Algorithm 2. Specifically, the main complexity comes from solving the standard SDP subproblems\footnote{For convex problems,  we assume that the Interior point method is adopted and then calculate the computational complexity accordingly \cite{boyd04}.}.
In terms of active beamforming design for Alice, the computational complexity on solving the subproblems \eqref{eq_trans_acti_wb} and \eqref{eq_activ_wc_trans} can be respectively calculated as $\mathcal{O}(M^{3.5})$ and $\mathcal{O}(M^{3.5})$. For Algorithm 1 on joint passive beamformeing design of STAR-RIS, the main complexity lies on solving the subproblem \eqref{eq_passi_trans_withoutrank1} which is dominated by $\mathcal{O}(2N^{3.5})$.
Besides, we use the bisection search method to find $\varphi(\epsilon)$  in \eqref{eq_trans_eq_passi_q} with the complexity of $\mathcal{O}\big(\operatorname{\log_2}(\frac{s_0}{\epsilon_\mathrm{t}})\big)$ where $s_0$ and $\epsilon_\mathrm{t}$  denote the length of the initial search interval and the accuracy tolerance, respectively. 
Hence, the computational complexity of Algorithm 1 
is $\mathcal{O}\big(I_2I_3(2N^{3.5})+\operatorname{\log_2}(\frac{s_0}{\epsilon_\mathrm{t}})\big)$, where $I_2$ and $I_3$ are the number of outer and inner iterations.
Based on the above analysis, the computational complexity of the proposed alternating Algorithm 2 for solving the original covert communication problem \eqref{eq_for_opt} is dominated by $\mathcal{O}\big(I_1\big(2M^{3.5}+\operatorname{\log_2}(\frac{s_0}{\epsilon_\mathrm{t}})+I_2I_3(2N^{3.5})\big)\big)$, where $I_1$ denotes the total iteration number of the proposed algorithm. The complexity is mainly determined by the number of antennas at Alice ($M$) and elements at STAR-RIS ($N$).

It is easy to verify that the convergence of the  alternating Algorithm 2. For each iteration of Algorithm 2, we can always find a solution not worse than that of the previous iteration, considering the SDR and the Dinkelbach's algorithm leveraged for solving the subproblems. Hence, the the objective function of problem \eqref{eq_for_opt} monotonically non-decreases w.r.t. the iteration index, and the algorithm finally converges subject to the transmit power limitation \eqref{eq_for_opt_1}, which will be further verified by simulation results.
\begin{center}
\begin{tabular}{p{8.5cm}}
\toprule[2pt]
\textbf{Algorithm 2:} Proposed Alternating Algorithm for STAR-RIS-assisted Covert Communications Problem \eqref{eq_for_opt} \\
\midrule[1pt]
1: Initialize feasible point ($\mathbf{w}_\mathrm{b}^{(0)}$, $\mathbf{w}_\mathrm{c}^{(0)}$, $\boldsymbol{\Theta}_\mathrm{r}^{(0)}$, $\boldsymbol{\Theta}_\mathrm{t}^{(0)}$); Define \\ \quad the tolerance accuracy  $\varepsilon$; Set iteration index $m=0$.\\
2: \textbf{While} $v>\varepsilon$ or $m=0$ \textbf{do}                       \\
3:\quad Solve the relaxed version of the subproblem \eqref{eq_trans_acti_wb} with \\
~~~~~SDR method and use Gaussian randomization method\\
~~~~~to construct the rank-one solution, then update the $\mathbf{w}_\mathrm{b}$. \\
4:\quad Similarly,  solve the relaxed version of the subproblem\\
\qquad \eqref{eq_activ_wc_trans} with SDR method and update the $\mathbf{w}_\mathrm{c}$. \\
5:\quad Solve subproblem \eqref{eq_passi} with Algorithm 1 and update\\
~~~~~the $\boldsymbol{\Theta}_\mathrm{r}$ and $\boldsymbol{\Theta}_\mathrm{t}$.\\
6:\quad Calculate the objective value $R_{\mathrm{bb}}^{(m+1)}$ value and update\\
~~~~~$v=\big|R_{\mathrm{bb}}^{(m+1)}-R_{\mathrm{bb}}^{(m)}\big|^2$; Let $m=m+1$.\\
7: \textbf{end while}      \\
\bottomrule[2pt]
\end{tabular}
\end{center}

\section{Simulation Results}\label{sec:S5}
In this section, we show the numerical simulation results to verify the effectiveness of the proposed STAR-RIS-assisted covert communication scheme implemented by the proposed optimization Algorithm 2. 
Specifically, all the simulation results are averaged over 1000 independent channels realizations.
For large-scale path loss coefficients, we assume $\rho_0=-20$dB, $\alpha=2.6$ and the distances are set as $d_\mathrm{AR}=500 m$, $d_\mathrm{rb}=100 m$, $d_\mathrm{rw}=80 m$ and $d_\mathrm{rc}=150 m$.
Furthermore, we define the noise power $\sigma_\mathrm{b}^2=-140$ dBm, $\sigma_\mathrm{c}^2=-140$ dBm and the self-interference coefficient $\phi=-160$ dB\cite{bharadia13}.
In the proposed algorithm, the accuracy tolerance parameters $\varepsilon$, $\varepsilon_1$ and $\varepsilon_2$ are set as $10^{-4}$,  $10^{-8}$ and  $10^{-8}$, respectively.

To highlight the advantage of covert communication aided by STAR-RIS, we consider a baseline scheme which employs two adjacent conventional RISs to replace STAR-RIS where one is the reflection-only RIS and the other one is transmission-only RIS. The number of elements in these two RISs is selected as $N/2$ so as to achieve a fair comparison. We call this baseline scheme as "RIS-aided scheme".
In addition, to further validate the effectiveness of the proposed algorithm, a general optimization method called GCMMA is adopted as a comparison algorithm to solve the problem \cite{svanberg07, svanberg02}, which requires lower computational complexity and is capable of converging to the Karush-Kuhn-Tucker (KKT) solution. 
%

\begin{figure}[ht]
\centering
\includegraphics[scale=0.19]{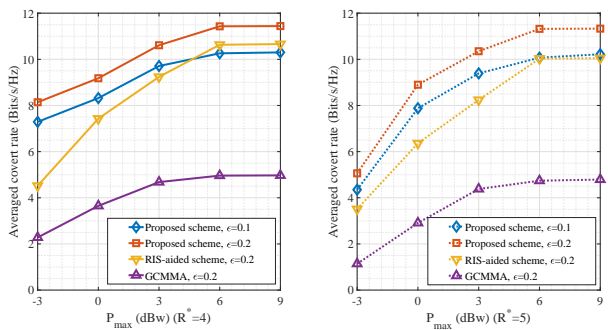}\\
\caption{Average covert rate versus the maximum transmit power $P_{\max}$ of Alice with $P_\mathrm{j}^{\max}=0$ dBw, $M=3$, $N=30$, $\iota=0.1$, $\kappa=0.1$, and different $\epsilon$ and $R^*$ .}\label{fig:Pvsrate}
\end{figure}
In Fig. \ref{fig:Pvsrate}, we show the performance of average covert rate versus the transmitted power $P_{\mathrm{max}}$ at Alice, considering different QoS ($R^*$) and covert ($\epsilon$) requirements.
Specifically,  we can find that the achievable covert rates for all schemes in all scenarios 
gradually increase with the growth of $P_{\max}$ before it reaching 6 dBw. And then the covert rates approach saturation when $P_{\max}$ further increases due to limitations of system settings. 
It is obvious that a significant performance improvement can be achieved by the proposed optimization scheme in comparison with the GCMMA method, which clearly validates the effectiveness of the proposed algorithm.
Compared with the RIS-aided baseline scheme, we can find that the proposed STAR-RIS-assisted scheme possesses a strong superiority in enhancing the system covet performance, and the advantage may further expanded when with relaxed QoS requirement but limited transmit power budget (smaller $R^*$ and $P_{\max}$). 
In addition, we can observe that a lower covert rate is achieved if the QoS or the covertness constraint becomes tighter, i.e., from $R^*=4$ to $R^*=5$, or from $\epsilon^*=0.2$ to $\epsilon^*=0.1$, which coincides with our intuition. Compared with the RIS-aided scheme, the performance degradation of proposed STAR-RIS assisted scheme with a moderate $P_{\max}$ is much less serious.

\begin{figure}[ht]
\centering
\includegraphics[scale=0.19]{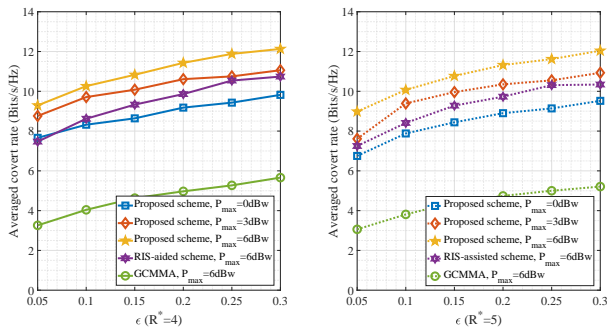}\\
\caption{Average covert rate versus the covert requirement $\epsilon$ with $P_\mathrm{j}^{\max}=0$ dBw, $M=3$, $N=30$, $\iota=0.1$, $\kappa=0.1$, and different $P_{\max}$ and $R^*$ .}\label{fig:epsilovsrate}
\end{figure}
Next in Fig. \ref{fig:epsilovsrate}, we investigate the influence of the covert requirement, i.e., $\epsilon$, on the performance of average covert rate, considering different $P_{\max}$ and QoS requirements. In particular, $P_{\max}=6$ dBw is selected to operate the RIS-aided baseline scheme for an evident comparison, and obvious performance improvement can be achieved by the proposed scheme. Even if a lower transmitted power budget, i.e., 3 dBw, is utilized, the proposed scheme can still obtain better performance. This is because the STAR-RIS possesses a more flexible regulation ability compared with the conventional RIS, which can adjust the element phases and amplitudes for both reflection and transmission. It can be seen that the proposed scheme highly outperforms the GCMMA algorithm and the performance gap enlarges with the increase of $\epsilon$, indicating that the proposed scheme can achieve a much better solution than the KKT solution converged by the GAMMA method. 
\begin{figure}[ht]
\centering
\includegraphics[scale=0.19]{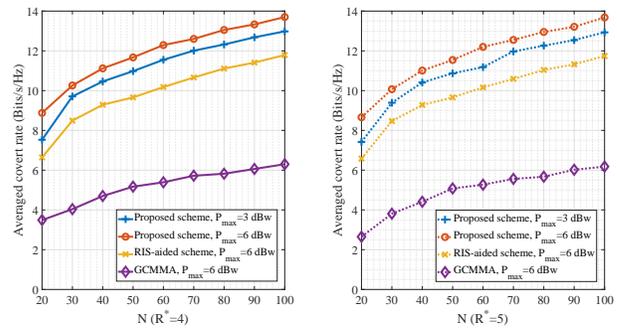}\\
\caption{Average covert rate versus the number of elements in STAR-RIS with $P_\mathrm{j}^{\max}=0$ dBw, $M=3$, $\epsilon=0.1$, $\iota=0.1$, $\kappa=0.1$, and different $P_{\max}$ and $R^*$ .}\label{fig:riselementvsrate}
\end{figure}

We present the variation curves of average covert rate w.r.t. the number of elements on STAR-RIS ($N$) in Fig. \ref{fig:riselementvsrate}, under different transmit power $P_{\max}$ and QoS constraints $R^*$. It can be observed that the average covert rates of all the schemes grow with $N$, since the increased elements can provide higher freedom degree for reconfiguration of propagation environment. However, the increasing rates gradually decrease with the growth of $N$, and this may due to the limitations of other system settings. Similarly, $P_{\max}=6$ dBw is chosen to implement the two benchmark schemes, i.e., the RIS-aided and the GCMMA schemes. 
The obtained results further verify the advantages of the proposed STAR-RIS-assisted scheme which can achieve even better performance than the benchmark schemes in the scenario with a much smaller transmit power budget ($P_{\max}=3$ dBw). In addition, the performance enhancement for the proposed scheme is more obvious as the number of elements becomes larger. 
\begin{figure}[ht]
\centering
\includegraphics[scale=0.37]{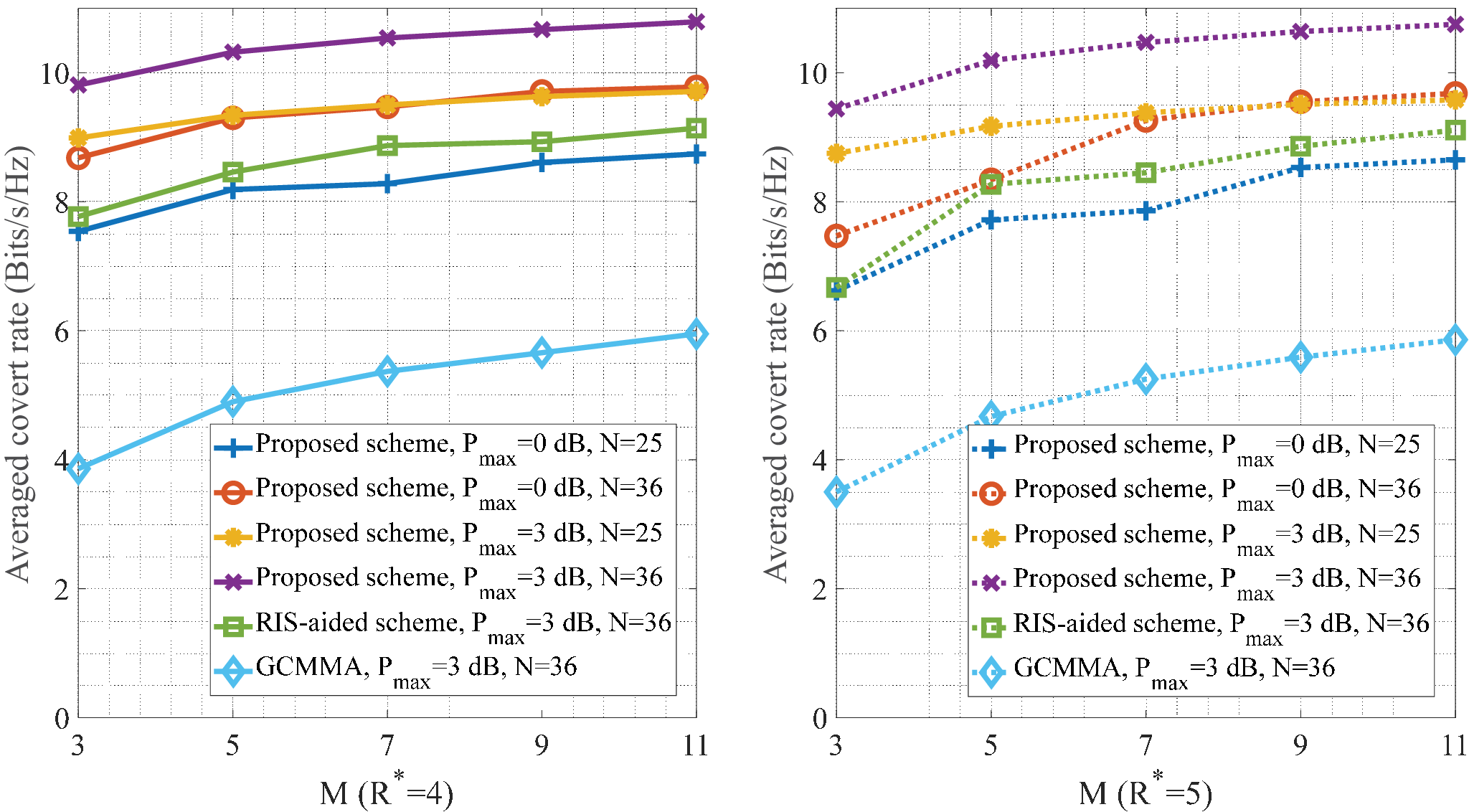}\\
\caption{Averaged covert rate versus the number of elements equipped at Alice with $P_\mathrm{j}^{\max}=0$ dBw, $\epsilon=0.1$, $\iota=0.1$, $\kappa=0.1$, and different $P_{\max}$, $N$ and $R^*$ .}\label{fig:antennavsrate}
\end{figure}

In Fig. \ref{fig:antennavsrate}, we explore the effects of the number of antennas equipped at the transmitter, i.e., $M$, on the available covert rate, considering different transmitted power $P_{\mathrm{max}}$, number of elements in STAR-RIS ($N$), and QoS requirements ($R^*$). In particular, with the growth of $M$, the average covert rates of all schemes gradually increase, but the increasing rates have downward trends. Besides, we can observe that higher transmitted power or more elements at STAR-RIS contribute to breaking the performance bottleneck imposed by channel characteristics and the number of antennas at Alice. Similarly, evident performance gaps exist between the proposed scheme and the RIS-aided scheme even for the proposed scheme with a lower $P_{\max}$ or fewer $N$. Compare with the GCMMA method, the proposed algorithm can still achieve better solutions in both scenarios with $R^*=4$ and $R^*=5$.
\begin{figure}[ht]
\centering
\includegraphics[scale=0.19]{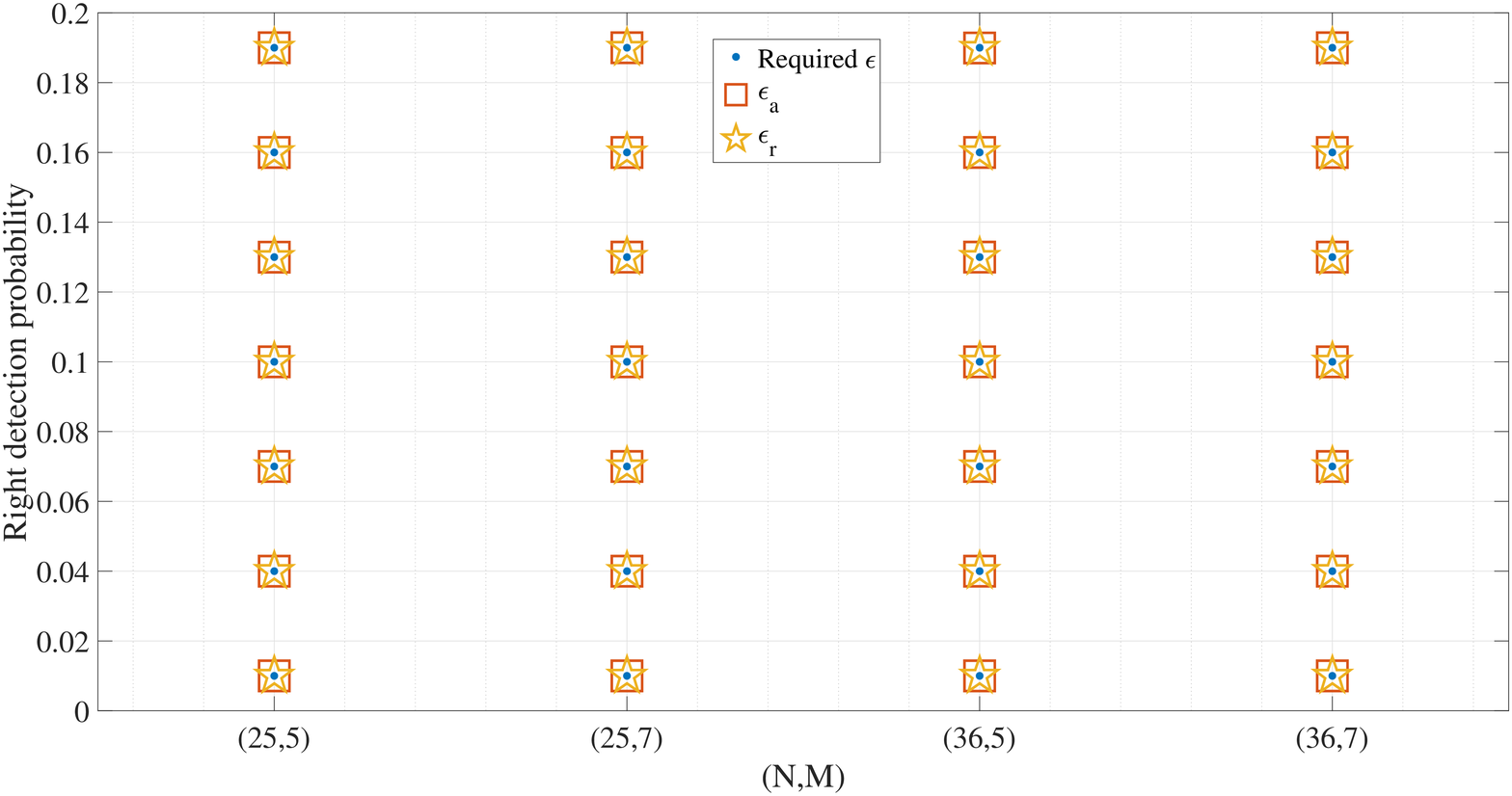}\\
\caption{Verifying the reasonability of choosing the lower bound of the minimum average DEP i.e., $\hat{P}_\mathrm{ea}^*$, with $P_\mathrm{j}^{\max}=0$ dBw, $P_{\max}=3$ dBw, $\iota=0.1$, $\kappa=0.1$, and different $N$ and $M$ .}\label{fig:verify}
\end{figure}

Fig. \ref{fig:verify} verifies the reasonability of adopting the lower bound of the minimum average DEP, i.e., $\hat{P}_\mathrm{ea}^*$,  to replace itself $\overline{P}_\mathrm{ea}^*$, under different covertness requirements ($\epsilon$), number of antennas at Alice ($M$), and elements equipped at RIS ($N$). In particular, we utilize the solution obtained by choosing the lower bound $\hat{P}_\mathrm{ea}^*$  to compute the accurate maximum average correct detection probability at Willie, which denotes $\epsilon_\mathrm{r}$. Note that $\epsilon_\mathrm{a}=1-\hat{\mathrm{P}}_\mathrm{ea}^*$ represents the upper bound of $\epsilon_\mathrm{r}$. According to the obtained results, we can find that the lower bound in \eqref{eq_low_avera_asy_DEP} is tight because $\epsilon_\mathrm{r}$ and $\epsilon_\mathrm{a}$ are almost identical in all considered scenarios with  different covert requirements and different communication system configurations.

\begin{figure*}[b]
\hrulefill
\begin{equation}\label{eq_PFA_detailed_derivation}
\begin{aligned}[b]
	P_{\mathrm{FA}}=&\operatorname{Pr}\left(\overline{P}_\mathrm{w}>\tau_\mathrm{d t} \mid \mathcal{H}_0\right)\\=&\begin{cases}
		1, & \tau_\mathrm{d t}<\sigma_\mathrm{w}^2 ,\\
		\int_0^{\infty} \frac{e^{-\frac{ x}{\lambda}}}{\lambda} d x-\int_0^{\tau_\mathrm{d t}-\sigma_\mathrm{w}^2} \int_0^{\tau_\mathrm{d t}-\sigma_\mathrm{w}^2-x} \frac{e^{-\frac{ x}{\lambda}}}{\lambda}  \frac{1}{\gamma P_\mathrm{j}^{\max}}d y d x,& \sigma_\mathrm{w}^2 \leq \tau_\mathrm{d t}<\sigma_\mathrm{w}^2+\gamma P_\mathrm{j}^{\max}, \\
		\int_0^{\gamma P_\mathrm{j}^{\max}} \int_{\tau_\mathrm{d t}-\sigma_\mathrm{w}^2-y}^{\infty}  \frac{e^{-\frac{ x}{\lambda}}}{\lambda}  \frac{1}{\gamma P_\mathrm{j}^{\max}} dx d y, & \tau_\mathrm{d t}\geq\gamma P_\mathrm{j}^{\max} +\sigma_\mathrm{w}^2,
	\end{cases} \\
	=&\begin{cases}
		1 ,& \tau_\mathrm{d t}<\sigma_\mathrm{w}^2 ,\\
		1-\frac{\left(\tau_\mathrm{d t}-\sigma_\mathrm{w}^2\right)+\lambda e^{-\frac{\tau_\mathrm{d t}-\sigma_\mathrm{w}^2}{\lambda}}-\lambda}{\gamma P_\mathrm{j}^{\max} }, & \sigma_\mathrm{w}^2 \leq \tau_\mathrm{d t}<\sigma_\mathrm{w}^2+\gamma P_\mathrm{j}^{\max}, \\				\frac{\mathrm{e}^{-\frac{\tau_\mathrm{d t}-\sigma_\mathrm{w}^2}{\lambda}}\left(e^{\frac{\gamma P_\mathrm{j}^{\max}}{\lambda}}-1\right) \lambda}{\gamma P_\mathrm{j}^{\max}}, &\tau_\mathrm{d t}\geq\gamma P_\mathrm{j}^{\max}+\sigma_\mathrm{w}^2,
	\end{cases}
\end{aligned}
\end{equation}
\end{figure*}

\vspace{-2mm}
\section{Conclusions}\label{sec:S6}
In this work, we initially investigate the application potentials of STAR-RIS in covert communications. In particular, the closed-form expression of the minimum DEP about the STAR-RIS-aided covert communication system is analytically derived. And then we jointly design the active and passive beamforming at the BS and STAR-RIS, to maximize the covert rate taking into account of the minimum DEP of Willie and the communication outage probability experienced at Bob and Carol.
Due to the strong coupling between active and passive beamforming variables and the characteristic amplitude constraint introduced by the STAR-RIS, the proposed optimization problem is a non-convex problem. To effectively solve this covert communication problem, we elaborately design an alternating algorithm based on the SDR method and Dinkelbach's algorithm. Simulation results demonstrate that the STAR-RIS-assisted covert communication scheme highly outperforms the conventional RIS-aided scheme. In addition, the proposed iterative algorithm can effectively solve the formulated problem with guaranteed convergence.

\appendices
\vspace{-3mm}
\section{Proof of Theorem 1}\label{Apd_A}
To facilitate the analysis, we denote $\varrho_1\triangleq\mathbf{h}_{\mathrm{rw}}^H \boldsymbol{\Theta}_\mathrm{r} \mathbf{H}_{\mathrm{AR}} \mathbf{w}_\mathrm{c}$, and it is easy to demonstrate that $\varrho_1$  follows the complex Gaussian distribution with mean zero and variance $\lambda=\left\|\mathbf{h}_{\mathrm{rw}}^H \boldsymbol{\Theta}_\mathrm{r}\right\|_2^2$ $\mathbf{w}_\mathrm{c}^H \mathbf{w}_\mathrm{c}$. Hence, the PDF of $|\varrho_1|^2$ can be written $f_{|\varrho_1|^2}(x)=\frac{e^{-\frac{ x}{\lambda}}}{\lambda}$. In addition, we know that $P_\mathrm{j}$ follows the uniform distribution, and thus the analytical expression of the FA probability $P_\mathrm{FA}$ can be derived as \eqref{eq_PFA_detailed_derivation}, which is shown at the bottom of this page. Similarly, we can derive the MD probability $P_\mathrm{MD}$ as \eqref{eq_P_MD}.

\vspace{-2mm}
\section{Proof of Theorem 2}\label{Apd_B}
According to the analytical expression of DEP at Willie, i.e., $P_\mathrm{e}$ given in \eqref{eq_DEP_expression}, we can see that $P_\mathrm{e}$ is a segment function based on the detection threshold $\tau_\mathrm{dt}$ in three different {ranges. 
Let us discuss the optimal detection $\tau_\mathrm{dt}^*$ threshold in three ranges, respectively.

\subsubsection{$\tau_\mathrm{d t}<\sigma_\mathrm{w}^2$}
It is easy to note that $P_\mathrm{e}=1$ when $\tau_\mathrm{d t}<\sigma_\mathrm{w}^2$, indicating that Willie is always unable to distinguish the existence of communications between Alice and Bob. Hence, there is no need to optimize  $\tau_\mathrm{dt}$ to minimize the DEP when $\tau_\mathrm{dt}$ falls into this range.

\subsubsection{$\sigma_\mathrm{w}^2\leq\tau_\mathrm{dt}<\sigma_\mathrm{w}^2+\gamma P_\mathrm{j}^{\max}$}
The first-order derivative of $P_\mathrm{e}$ w.r.t. $\tau_\mathrm{dt}$ in this range is given by
\begin{equation}\label{eq_partial_tua_range1}
\frac{\partial P_\mathrm{e}}{\partial \tau_\mathrm{d t}}=\frac{e^{-\frac{\tau_\mathrm{d t}-\sigma_\mathrm{w}^2}{\lambda}}-e^{-\frac{\tau_\mathrm{d t}-\sigma_\mathrm{w}^2}{\tilde{\lambda}}}}{\gamma P_\mathrm{j}^{\max}},
\end{equation}
from which can find that $\frac{\partial P_\mathrm{e}}{\partial \tau_\mathrm{d t}}<0$ always holds. Therefore, $P_\mathrm{e}$ decrease monotonically versus $\tau_\mathrm{d t}\in[\sigma_\mathrm{w}^2, \sigma_\mathrm{w}^2+\gamma P_\mathrm{j}^{\max})$ and the optimal detection threshold is $\tau_\mathrm{d t}^*=\sigma_\mathrm{w}^2+\gamma P_\mathrm{j}^{\max}$.

\subsubsection{$\tau_\mathrm{d t}\geq\sigma_\mathrm{w}^2+\gamma P_\mathrm{j}^{\max}$}
We can further obtain the first-order derivative of $P_\mathrm{e}$ w.r.t. $\tau_\mathrm{dt}$  in this range as
\begin{equation}\label{eq_partial_tua_range2}
\frac{\partial P_\mathrm{e}}{\partial \tau_\mathrm{d t}}=\frac{e^{-\frac{\tau_\mathrm{d t}-\sigma_\mathrm{w}^2}{\tilde{\lambda}}}\left(e^{\frac{\gamma P_\mathrm{j}^{\max}}{\tilde{\lambda}}}-1\right)+e^{-\frac{\tau_\mathrm{d t}-\sigma_\mathrm{w}^2}{\lambda}}\left(1-e^{\frac{\gamma P_\mathrm{j}^{\max}}{\lambda}}\right)}{\gamma P_\mathrm{j}^{\max}}.
\end{equation}

Let $\frac{\partial P_\mathrm{e}}{\partial \tau_\mathrm{d t}}=0$, we can obtain the unique solution of $\tau_\mathrm{d t}$ in this range, i.e.,  $\tau_\mathrm{d t}=\frac{\tilde{\lambda} \lambda}{\tilde{\lambda}-\lambda}\ln\Delta+\sigma_\mathrm{w}^2 \in[\sigma_\mathrm{w}^2+\gamma P_\mathrm{j}^{\max},+\infty)$ where $\Delta=\frac{e^{\frac{\gamma P_\mathrm{j}^{\max} }{\lambda}}-1}{e^{\frac{\gamma P_\mathrm{j}^{\max} }{\tilde{\lambda}}}-1}$.
It is easy to prove that $P_\mathrm{e}$ first decreases and then increases versus  $\tau_\mathrm{d t}$  in this range with $\frac{\tilde{\lambda} \lambda}{\tilde{\lambda}-\lambda} \ln \Delta+\sigma_\mathrm{w}^2$ as the the inflection point. Hence, the optimal detection threshold for minimizing $P_\mathrm{e}$ is given as $\tau_\mathrm{d t}^*=\frac{\tilde{\lambda} \lambda}{\tilde{\lambda}-\lambda} \ln \Delta+\sigma_\mathrm{w}^2$.

Based on the above analysis, the optimal detection threshold $\tau_\mathrm{d t}^*$ can be finally expressed as
\begin{equation}\label{eq_opti_tau}
\tau_\mathrm{d t}^*=\begin{cases}
	\sigma_\mathrm{w}^2+\gamma P_\mathrm{j}^{\max} , & \sigma_\mathrm{w}^2 \leq \tau_\mathrm{d t}<\sigma_\mathrm{w}^2+\gamma P_\mathrm{j}^{\max}, \\
	\frac{\tilde{\lambda} \lambda}{\tilde{\lambda}-\lambda} \ln \Delta+\sigma_\mathrm{w}^2. & \tau_\mathrm{d t}\geq\gamma P_\mathrm{j}^{\max}+\sigma_\mathrm{w}^2.
\end{cases}
\end{equation}

We can verify that $P_\mathrm{e}$ in \eqref{eq_DEP_expression} is a continuous segment function at the segment points  $\sigma_\mathrm{w}^2$ and $\sigma_\mathrm{w}^2+\gamma P_\mathrm{j}^{\max}$. 
Therefore, the optimal detection threshold in the overall defined region for minimizing the DEP $P_\mathrm{e}$ is $\tau_\mathrm{d t}^*=\frac{\tilde{\lambda} \lambda}{\tilde{\lambda}-\lambda} \ln \Delta+\sigma_\mathrm{w}^2$.

\vspace{-2mm}
\section{Proof of Theorem 3}\label{Apd_C}
When the required transmission rate between Alice and Bob is chosen as $R_b$, the communication outage probability at Bob under the randomness of the jamming power $P_\mathrm{j}$ can be calculated as
\begin{equation}\label{eq_outage_detailed_derivation_b}
\begin{aligned}[b]
	\delta_{\mathrm{AB}}=&\operatorname{Pr}\left(C_\mathrm{b}<R_\mathrm{b}\right)\\
	=&\operatorname{Pr}\left(P_\mathrm{j}>\Upsilon\right)\\
	=&\begin{cases}
		1, & \Upsilon<0,\\
		\int_{\Upsilon}^{P_\mathrm{j}^{\max}}\frac{1}{P_\mathrm{j}^{\max}}d y, & 0\leq\Upsilon<P_\mathrm{j}^{\max},\\
		0, &\Upsilon\geq P_\mathrm{j}^{\max},
	\end{cases}\\
	=&\begin{cases}
		1, & \Upsilon<0,\\
		1-\frac{\Upsilon}{P_\mathrm{j}^{\max}}, & 0\leq\Upsilon<P_\mathrm{j}^{\max},\\
		0, &\Upsilon\geq P_\mathrm{j}^{\max},
	\end{cases}
\end{aligned}
\end{equation}
where  $\Upsilon=\frac{\left|\mathbf{h}_{\mathrm{rb}}^H \boldsymbol{\Theta}_\mathrm{r} \mathbf{H}_{\mathrm{AR}} \mathbf{w}_\mathrm{b}\right|^2-\left(2^{R_\mathrm{b}}-1\right)\left(\left|\mathbf{h}_{\mathrm{rb}}^H \boldsymbol{\Theta}_\mathrm{r} \mathbf{H}_{\mathrm{AR}} \mathbf{w}_ \mathrm{c}\right|^2+\sigma_\mathrm{b}^2\right)}{\left(2^{R_\mathrm{b}}-1\right)\left|\mathbf{h}_{\mathrm{rb}}^H \boldsymbol{\Theta}_\mathrm{t} \mathbf{h}_{\mathrm{rc}}^*\right|^2}$.

Similarly, when the required transmission rate between Alice and Carol is chosen as $R_c$, the communication outage probability at Carol under the randomness of the jamming power $P_\mathrm{j}$ and the self-interference channel $h_{\mathrm{cc}}$ is derived as
\begin{equation}\label{eq_outage_detailed_derivation_c}
\begin{aligned}[b]
	\delta_{\mathrm{AC}}=&\operatorname{Pr}\left(C_\mathrm{c}<R_\mathrm{c}\right)\\
	=&\operatorname{Pr}\left(|g_\mathrm{cc}|^2P_\mathrm{j}>\frac{\Gamma}{\phi}\right)\\
	=&\begin{cases}
		\int_{0}^{P_\mathrm{j}^{\max}}\int_{\frac{\Gamma}{\phi y}}^{+\infty}e^{-x}\frac{1}{P_\mathrm{j}^{\max}}d x d y	, & \Gamma\geq0,\\
		1, & \Gamma<0,\\
	\end{cases}\\
	=&\begin{cases}
		e^{-\frac{\Gamma}{\phi P_\mathrm{j}^{\max}}}+\frac{\Gamma}{\phi P_\mathrm{j}^{\max}}\operatorname{Ei}\left(-\frac{\Gamma}{\phi P_\mathrm{j}^{\max}}\right), & \Gamma\geq0,\\
		1, & \Gamma<0,\\
	\end{cases}
\end{aligned}
\end{equation}
where $\Gamma=\frac{\left|\mathbf{h}_{\mathrm{rc}}^H\boldsymbol{\Theta}_\mathrm{t} \mathbf{H}_{\mathrm{AR}} \mathbf{w}_ \mathrm{c}\right|^2-\left(2^{R_\mathrm{c}}-1\right)\left(\left|\mathbf{h}_{\mathrm{rc}}^H\boldsymbol{\Theta}_\mathrm{t} \mathbf{H}_{\mathrm{AR}} \mathbf{w}_ \mathrm{b}\right|^2+\sigma_\mathrm{c}^2\right)}{\left(2^{R_\mathrm{c}}-1\right)}$. 

\ifCLASSOPTIONcaptionsoff 
\newpage
\fi

\bibliographystyle{IEEEtran}
\bibliography{CC}

\end{document}